\UseRawInputEncoding
\pdfoutput=1
\documentclass[10pt,pra,aps,twocolumn]{revtex4-1}
\usepackage{amsmath,bbm}
\usepackage{amsthm}
\usepackage{latexsym}
\usepackage{amssymb}
\usepackage{float}
\usepackage{graphicx}           
\usepackage{color}
\usepackage{mathpazo}
\usepackage{relsize}
\usepackage{braket}
\usepackage[colorlinks=true,linkcolor=blue,citecolor=magenta,urlcolor=blue]{hyperref}
\usepackage{changes}
\usepackage{cancel}
\usepackage{xcolor}

\usepackage[utf8]{inputenc}
\usepackage[T1]{fontenc}

\newcommand{\be}{\begin{equation}} 
\newcommand{\ee}{\end{equation}}
\newcommand{\beq}{\begin{eqnarray}}
\newcommand{\eeq}{\end{eqnarray}}

\def\squareforqed{\hbox{\rlap{$\sqcap$}$\sqcup$}}
\def\qed{\ifmmode\squareforqed\else{\unskip\nobreak\hfil
\penalty50\hskip1em\null\nobreak\hfil\squareforqed
\parfillskip=0pt\finalhyphendemerits=0\endgraf}\fi}
\def\endenv{\ifmmode\;\else{\unskip\nobreak\hfil
\penalty50\hskip1em\null\nobreak\hfil\;
\parfillskip=0pt\finalhyphendemerits=0\endgraf}\fi}
\DeclareMathOperator{\Tr}{Tr}
\newcommand{\I}{\mathbbm{1}}

\newcommand{\ra}{\rangle}
\newcommand{\la}{\langle}

\makeatletter
\newtheorem*{rep@theorem}{\rep@title}
\newcommand{\newreptheorem}[2]{%
\newenvironment{rep#1}[1]{%
 \def\rep@title{#2 \ref{##1}}%
 \begin{rep@theorem}}%
 {\end{rep@theorem}}}
\makeatother

\def\tr{\mbox{tr}}
\theoremstyle{remark}
\newtheorem{thm}{Theorem}
\newreptheorem{thm}{Theorem}

\theoremstyle{remark}
\newtheorem{lemma}{Lemma}

\newtheorem{defi}{Definition}

\begin{document}

\title{Certification of two-qubit quantum systems with temporal inequality}

\author{Chellasamy Jebarathinam}
\author{Gautam Sharma}
\author{Sk Sazim}
\author{Remigiusz Augusiak}
\affiliation{Center for Theoretical Physics, Polish Academy of Sciences, Aleja Lotnik\'{o}w 32/46, 02-668 Warsaw, Poland}

\begin{abstract}

Self-testing of quantum devices based on observed measurement statistics  
is a method to certify quantum systems using minimal resources.
In Ref. [Phys. Rev. \textbf{A} 101, 032106 (2020)], a scheme based on observing measurement statistics that demonstrate Kochen-Specker contextuality has been shown 
to certify two-qubit entangled states and measurements without the 
requirement of spatial separation between the subsystems. However, this scheme assumes a set of compatibility conditions on the measurements 
which are crucial to demonstrating Kochen-Specker contextuality. In this work, we propose a self-testing protocol to certify the above two-qubit states and measurements without the assumption of the compatibility conditions, and at the same time without requiring the spatial separation between the subsystems. Our protocol is based on the observation of sequential correlations leading to the maximal violation of a temporal inequality derived from non-contextuality inequality. Moreover, our protocol is robust to small experimental errors or noise.
\end{abstract}

\maketitle
\section{Introduction}
Realizing fault-tolerant quantum computation plays a key role in achieving quantum technologies. A fault-tolerant quantum computer may be achieved by making use of topological qubits in which qubits are encoded as the subspace of a high dimensional space of suitable physical system such as Majorana fermions \cite{NSS+08,KKL+17}. The circuit model of quantum computing can be replaced by Pauli-based computing which makes use of only a minimal number of Pauli measurements \cite{BSS16}. Such Pauli-based computing can be realized by topological qubits \cite{MB22}.  

Before performing any quantum information processing task, it is important to certify how close are the relevant quantum devices to the ideal ones. There are the certification methods such as tomography \cite{PCZ97,DL01,TNW+02} and self-testing \cite{SB20} that have been explored to certify the relevant quantum devices. While certifying quantum devices, it is desirable to use methods that make use of a limited number of resources and minimize the assumptions on the quantum devices. Certification schemes based on tomography can become resource-consuming depending on the quantum system to be certified and it requires certain assumptions such as the Hilbert space dimension of the quantum system has to be trusted.  
On the other hand, self-testing of quantum devices aims to provide certification using a minimal number of resources and a minimal number of assumptions. Self-testing was originally introduced as a device-independent certification of quantum state and measurements for quantum cryptographic applications \cite{MY04}. Device-independent certification schemes are based on the observation of non-classical correlations that imply Bell non-locality \cite{Bel64, BCP+14}. These schemes do not depend on the detailed characterization of the quantum devices, importantly, they do not assume the Hilbert space dimension of the quantum systems to be certified.

Self-testing based on  Bell non-locality has a restriction that it can be used to certify composite quantum systems that have entanglement and admit spatial separation between the subsystems and local measurements on subsystems. Recently, self-testing based on observation of Kochen-Specker contextuality \cite{KS67} has been explored to certify quantum systems which do not have entanglement or  do admit a spatial separation between the subsystems \cite{BRV+19, BRV+21,IMO+20,SJA22}.  Such a certification method is relevant in the context of certification quantum devices for quantum computation purposes as computation typically happens in a local fashion.   
Since quantum measurements that demonstrate Kochen-Specker contextuality
require to satisfy certain compatibility conditions, and self-testing statements based on contextuality explicitly require the assumption of certain compatibility conditions. 

Sequential correlation inequalities such as Legget-Garg inequality \cite{LG85} or temporal inequality \cite{MKT+14} can be used to demonstrate non-classicality based on quantum measurements on single quantum systems without the assumption of certain compatibility conditions. Motivated by this, such inequalities have also been explored to certify quantum measurements on single quantum systems without the requirement of any compatibility conditions on the quantum measurements \cite{MMJ+21,SSA20,DMS+22,SKB22}. However these certification schemes also make a few minimal assumptions, for instance, in Ref. \cite{MMJ+21,SSA20,DMS+22} it is assumed that a) the preparation device always prepares a maximally mixed state and  b) measurements do not have any memory and they return only the post-measurement state, while in Ref. \cite{SKB22} one can certify only the temporal correlation matrix. Moreover, one can lift the assumption of producing maximally mixed states by the measurement device with a) ensuring that the input state to the device is always a full rank state and b) sequential action of same measurement will output a certain outcome always \cite{2023arXiv230701333S}.

In Ref. \cite{IMO+20}, a certification scheme  has been proposed to certify quantum measurements of a fermionic system for topological quantum computing. These quantum measurements are a subset of two-qubit Pauli measurements that demonstrate Kochen-Specker contextuality. As this scheme is based on observed statistics that imply Kochen-Specker contextuality, it assumes certain compatibility conditions on the quantum measurements. In this work, we  propose a certification scheme based on a temporal inequality to certify  
the subset of two-qubit Pauli measurements along with the quantum state. Our scheme may be applied to certify those two-qubit quantum states and measurements for topological quantum computing without assuming the compatibility relations, unlike the approach of \cite{IMO+20}. However, we assume that the measurements are projective. In this respect, we would like to comment that the Ref. \cite{DMS+22} assumes that measurement device produces always maximally mixed state as well as same measurement applied many times gives similar outcome in order to prove that measurements involved in the experiments are projective. Also, a similar conclusion can be achieved by assuming full rank input to the measurement device instead of maximally mixed output state \cite{2023arXiv230701333S}. Therefore, the above three assumptions seem equivalent to us.  
Presently, we do not have clear understanding which one is more practical with respect to experimental point of view. Moreover, we also show that our self-testing scheme is robust to small experimental errors or noise.       

Before we discuss our methodology and main results, we provide some background information necessary for further considerations.

\textit{Sequential correlations.}--We consider a sequential measurements scenario to observe a violation of a temporal inequality. The scenario consists of a preparation device that prepares a single quantum system in an unknown quantum state $\rho$ and a sequence of two and three black-box operations. These operations are dichotomic measurements denoted by $A_i$, with $i=1,\ldots,6$. The outcomes of these measurements are denoted by $a_i$ which can take the values $\pm 1$. We then assume that these measurements are projective, and so $A_i$ are the standard quantum observables such that $A_i^2=\mathbbm{1}$. Notice here that, while we do not restrict the dimension of the underlying Hilbert space, we cannot, however, exploit here the Naimark dilation argument that allows one to represent generalized measurements in terms of projective ones acting on a higher-dimensional Hilbert space. This is because we do not impose any commutation relations for the measurements.

In our sequential scenario, we finally need to assume that the measurements $A_i$ are non-demolishing, i.e., they do not destroy the physical system, a also that each black box has no memory and returns the actual post-measurement state.

In an experiment that realizes this scenario, the joint probabilities corresponding to the sequences of two measurements and three measurements can be observed. For instance, the joint probabilities of measuring $\{A_1,A_4\}$ in the temporal order  $A_1 \rightarrow  A_4$ and  $\{A_1,A_2,A_3\}$  in the temporal order $A_1 \rightarrow  A_2  \rightarrow  A_3$ is given by $P(a_1,a_4|A_1,A_4)$ and $P(a_1,a_2,a_3|A_1,A_2,A_3)$, respectively. 
In terms of these joint probabilities, the second and third-order sequential correlations denoted by 
$\braket{A_1A_4}_{\pi}$ and $\braket{A_1A_2A_3}_{\pi}$ are defined as 
\begin{align*}
\braket{A_1A_4}_{\pi}=\sum_{a_1,a_4}a_1a_4 P(a_1,a_4|A_1,A_4),\:\: {\rm and}~~~~~~~~~\\
\braket{A_1A_2A_3}_{\pi}=\sum_{a_1,a_2,a_3}a_1a_2a_3 P(a_1,a_2,a_3|A_1,A_2,A_3), 
\end{align*}
respectively.

%
Quantum mechanical version of the  sequential correlations, for instance, 
such as $\braket{A_1A_4}_{\pi}$ and
$\braket{A_1A_2A_3}_{\pi}$ can then be written as,   
%
\begin{align}\label{dupablada}
    \la A_1 A_4 \ra _{\pi} = \frac{1}{2}\tr \bigg(\rho  \{A_1,A_4 \} \bigg),\:\: {\rm and}~~~~~~~~~\\
    \la A_1 A_2 A_3 \ra _{\pi}  = \frac{1}{4} \tr \bigg(\rho  \{ A_1,\{ A_2,A_3 \} \} \bigg),
\end{align}
respectively, where $\rho=\ket{\psi}\!\!\bra{\psi}$ (see Appendix in Ref. \cite{PQK18} for the derivation of such formulae) and $\{A,B\}=AB+BA$ is the standard anti-commutator. The other sequential correlations can also be similarly expressed in terms of the anti-commutator.

\textit{Non-contextuality inequality.}--Using these sequential correlations one can demonstrate both contextual and temporal correlations. We will now first describe how to demonstrate contextual correlation here. Suppose the sets of measurements $\{A_1,A_4\}$, $\{A_2,A_5\}$ and $\{A_3,A_6\}$, $\{A_1,A_2,A_3\}$ and $\{A_4,A_5,A_6\}$ that are measured sequentially are treated as contexts, i.e., the measurements in these sets commute with each other:
\begin{align}\label{dupa}
&    [A_1,A_4]=[A_2,A_5]=[A_3,A_6]=[A_1,A_2]=[A_1,A_3] \nonumber \\
& =[A_2,A_3]=[A_4,A_5]=[A_4,A_6]=[A_5,A_6]=0,
\end{align}
where $[A,B]=AB-BA$. 
Then, the quantum mechanical version of the sequential correlations simplify as $\braket{A_1A_4}_{\pi}=\braket{A_1A_4}$ and $\braket{A_1A_2A_3}_{\pi}=\braket{A_1A_2A_3}$, such that these sets of measurements can be used to demonstrate Kochen-Specker contextuality. For instance, the sequential
measurements of these contexts with the following two-qubit observables:
\beq \label{Ai}
A_1 &=& X \otimes \mathbbm{1}, \qquad A_4 =  \mathbbm{1} \otimes X, \nonumber \\
A_2 &=& \mathbbm{1} \otimes Z, \qquad A_5 =  Z \otimes \mathbbm{1}, \nonumber \\
A_3 &=&  X \otimes Z, \qquad          A_6 =  Z  \otimes X
\eeq
on the two-qubit maximally entangled state,
\begin{equation}\label{2qmes}
\ket{\phi^+}=(\ket{00}+\ket{11})/\sqrt{2},
\end{equation}
imply Kochen-Specker contextuality through a logical contradiction
by providing noncontextual value assignments to each observable in the context. The above contexts form a subset of the contexts pertaining to 
the famous Peres-Mermin square demonstrating Kochen-Specker contextuality 
for any two-qubit state \cite{Per90, Mer90}.
Moreover, one can consider a non-contextuality inequality as used in Ref. \cite{CFR+08}  in terms of the  expectation values corresponding to joint measurements of the observables in the above contexts  as follows: 
\begin{align}\label{ncIneq}
    &\mathcal{I}_{NC}=\braket{A_1A_2A_3}+\braket{A_4A_5A_6} \nonumber  \\
    &+\braket{A_1A_4}+\braket{A_2A_5}-\braket{A_3A_6} \le \eta_C,
\end{align}
where $\eta_C=3$ is the classical bound of the inequality. The quantum bound $\eta_Q=5$ of the above inequality can be obtained from the state given by Eq. (\ref{2qmes})
for the measurements given by Eq. (\ref{Ai}). In Ref. \cite{IMO+20}, the measurement statistics pertaining to the maximal violation of the inequality (\ref{ncIneq}) have been shown to provide robust certification of two-qubit Pauli measurements in a set-up such as Majorana Fermions.

Now, the stage is set to discuss our set up and main results regarding certification of two-qubit state and measurements using a temporal inequality. In the next section (Sec. \ref{self-testing2qPm}), we discuss the procedure to obtain temporal inequality from the non-contextuality inequality considered earlier within the sequential correlations scenario. We also discuss how such inequality can be used to certify the two qubit systems with permissible small experimental errors and noises. We conclude in the Sec. \ref{Conc}.

\section{Robust certification of the two-qubit systems using temporal inequality}\label{self-testing2qPm}
\subsection{temporal inequality}
Our goal is to propose a temporal extension of the inequality (\ref{ncIneq}) to certify the two-qubit maximally entangled state \eqref{2qmes} and measurements \eqref{Ai} without assuming the compatibility relations (\ref{dupa}). To this aim, we consider the following temporal inequality:
\begin{widetext}
   \begin{align}\label{TncIneq}
   \mathcal{I}_T :=\frac{1}{2}\Big(\braket{A_1A_2A_3}_{\pi}+\braket{A_2A_1A_3}_{\pi}+\braket{A_4A_5A_6}_{\pi} +\braket{A_5A_4A_6}_{\pi}  \Big)  +\braket{A_1A_4}_{\pi}+\braket{A_2A_5}_{\pi}-\braket{A_3A_6}_{\pi} \le \eta_C,
\end{align} 
\end{widetext}
where no compatibility relations for the measurements are assumed.
It has been argued in Ref. \cite{MKT+14} that the existence of a single joint probability distribution for all the outcomes of the measurements as the common underpinning assumption for the Bell-type,
Kochen-Specker non-contextuality and temporal inequalities. Therefore, one can derive 
the classical bound of any of these three types of inequalities as the maximum value that can be attained by any deterministic strategy. Hence, for the inequality \eqref{TncIneq} we have
\begin{eqnarray}
\eta_C &\!\!=\!\!& \max_{a_i\in\{ \pm 1\}} [  \frac{1}{2}(a_1a_2a_3+a_2a_1a_3 +a_4a_5a_6+a_5a_4a_6)  \nonumber \\
&& \quad\quad\quad\quad\quad+a_1a_4+a_2a_5-a_3a_6 ],\nonumber\\
&\!\!=\!\!&\max_{a_i\in\{ \pm 1\}} [   a_1a_2a_3 +a_4a_5a_6   
+a_1a_4+a_2a_5-a_3a_6 ]\nonumber
\end{eqnarray}
which gives $\eta_C=3$. 

Then, in order to determine the maximal quantum value $\eta_Q$ of $\mathcal{I}_T$ let us observe that it directly follows from the definitions \eqref{dupablada} that the sequential expectation values are bounded as 
\begin{equation*}
 -1\leq\braket{A_iA_j}_{\pi}\leq 1,   \quad
    -1\leq\braket{A_iA_jA_k}_{\pi}\leq 1
\end{equation*}
for any $i$, $j$ and $k$. As a consequence
$\eta_Q=5$. What is more, this value is achieved by the same quantum state 
and measurements that lead to the maximal quantum violation of the
non-contextuality inequality \eqref{ncIneq}.
In the next section, we demonstrate that the quantum bound of the inequality (\ref{TncIneq}) can be used to certify this quantum realization without the assumption of the compatibility relations as used in the case of the inequality (\ref{ncIneq}). 

\subsection{Robust self testing of two-qubit systems}
Consider an experimental situation that realizes the sequential scenario, in which we assume that both the state and measurements and their compatibility relations are unknown. 
Together, we consider a reference experiment 
with a known pure state $\ket{\hat{\psi}}\in\mathbbm{C}^d$ for 
$d=4$ and known observables $\hat{A}_i$ acting on $\mathbbm{C}^4$
that obey the compatibility relations given by Eq. (\ref{dupa}). We assume that our reference experiment leads to the maximal violation of the 
temporal inequality is given by Eq. (\ref{TncIneq}).
With these two experiments at hand, for our purpose, we adopt the definition of self-testing used in Ref. \cite{IMO+20} 
(see also Ref. \cite{SSA20}) as given below. 
\begin{defi}\label{defselftesting}
Given that the state $|\psi \ra  \in \mathcal{H}$ and a set of measurements $ A_i $ with $i\in [1,6]$, maximally violate the given temporal inequality (\ref{TncIneq}), then the self testing of the state $| \hat{\psi}  \ra  \in \mathbbm{C}^4$ and the set of observables $ \hat{A_i} $ acting on  $\mathbbm{C}^4$ is defined by the existence of a projection $P:\mathcal{H} \rightarrow \mathbb{C}^4$ and a unitary $U$ acting  on $\mathbbm{C}^4$ such that
\beq\label{DefSelf}
U^\dagger (P\, A_i\, P^\dagger) U &=& \hat{A}_i, \\
U (P\,|\psi \ra) &=& |\hat{\psi} \ra.
\eeq
\end{defi}
In other words, the above definition implies 
that based on the observed maximal violation of the inequality, 
one is able to identify 
a subspace $V=\mathbbm{C}^4$ in $\mathcal{H}$ on which 
all the observables act invariantly. Equivalently, 
$A_i$s can be decomposed as $A_i=\hat{A}_i\oplus A_{i}'$, where
$\hat{A}_i$ act on $V$, whereas $A_i'$ act on the orthocomplement of $V$ in $\mathcal{H}$; in particular, $A_i'\ket{\psi}=0$. Moreover, there is a unitary 
$U^{\dagger}\,\hat{A}_i\,U=\hat{A}_i$.

\textit{Symmetries of the inequality.} Before we proceed with the self-testing proof, let us observe the following symmetries of the inequality (\ref{TncIneq})
\begin{align} 
    &1.\quad  A_1 \leftrightarrow A_2 \quad \text{and} \quad A_4  \leftrightarrow A_5  \label{symmetry1} \\
    &2. \quad A_1 \leftrightarrow A_3\quad \text{and} \quad A_4 \leftrightarrow -A_6 \label{symmetry2}\\
    &3. \quad A_2 \leftrightarrow -A_3 \quad \text{and}\quad A_5 \leftrightarrow A_6 \label{symmetry3}\\
    &4. \quad A_1 \leftrightarrow A_4 \quad \text{and} \quad A_2 \leftrightarrow A_5 \quad \text{and} \quad A_3 \leftrightarrow A_6 \label{symmetry4}
\end{align}
It can be checked that in any of these four cases, the inequality  (\ref{TncIneq}) remains invariant. In the following, we will employ these symmetries to prove our self-testing statement.

Now, assume that our first experiment achieves the quantum bound of the inequality \eqref{TncIneq}. 
This directly implies that all except the last correlation functions of the inequality
take value $1$, whereas the last term equals $-1$, which
via the Cauchy-Schwartz inequality translate to the following equations:
\beq
A_1A_2A_3 |\psi\ra &=& |\psi\ra \quad \& \quad \texttt{permutations}, \label{MaxQ_const1} \\
A_4A_5A_6 |\psi\ra &=& |\psi\ra  \quad \& \quad \texttt{permutations},  \label{MaxQ_const2} \\
A_1A_4 |\psi\ra &=& |\psi\ra  \quad \& \quad \texttt{permutations},    \label{MaxQ_const3}  \\ 
A_2A_5 |\psi\ra &=& |\psi\ra  \quad \& \quad \texttt{permutations},  \label{MaxQ_const4} \\
A_3A_6 |\psi\ra &=& -|\psi\ra  \quad \& \quad \texttt{permutations} \label{MaxQ_const5} 
\eeq
where \texttt{permutations} refers to the fact that the above relations also hold if we permute the observables. Using these identities,
we now proceed to show the following anti-commutation relations:
\begin{align}\label{anticom}
&\{A_1,A_5\}|\psi\ra = \{A_1,A_6\}|\psi\ra=\{A_2,A_4\}|\psi\ra
\nonumber \\ &= \{A_2,A_6\}|\psi\ra  
=\{A_3,A_4\}|\psi\ra = \{A_3,A_5\}|\psi\ra= 0.
\end{align}

One directly deduces from the identities given by  Eqs. (\ref{MaxQ_const1})-(\ref{MaxQ_const5}) that
\beq \label{preSym2q}
A_1A_5 |\psi\ra &=A_1A_2 |\psi\ra = A_3 |\psi\ra 
                =-A_6 \ket{\psi} \nonumber \\&= - A_5 A_4 |\psi\ra = - A_5A_1 |\psi\ra, 
\eeq
where in the first line we used $A_5 |\psi\ra = A_2 |\psi\ra$ from Eq. \eqref{MaxQ_const4} and  $A_1A_2 |\psi\ra = A_3|\psi\ra$ that stems from Eq. \eqref{MaxQ_const1}. On the other hand, in the second line, we used $A_3 |\psi\ra = -A_6 |\psi\ra$ from Eq. \eqref{MaxQ_const5}, 
 $A_6 |\psi\ra =A_5A_4 |\psi\ra$ from Eq. \eqref{MaxQ_const2}
and $A_4 |\psi\ra = A_1|\psi\ra$ from Eq. \eqref{MaxQ_const3}. 

 In the following, we employ the symmetries in Eqs. \eqref{symmetry1}-\eqref{symmetry4} to obtain other similar relations. Employing the first, second and third symmetries \eqref{symmetry1}-\eqref{symmetry3} of the inequality in the argument given by Eq. (\ref{preSym2q}), we can get 
\beq \label{preSym2qa2}
A_2A_4 |\psi\ra &=& -A_4A_2|\psi\ra, \label{preSym2qa1} \\ 
A_3A_5 |\psi\ra &=& -A_5A_3|\psi\ra, \label{preSym2qa2} \\
A_1A_6 |\psi\ra &=& -A_6A_1|\psi\ra, \label{preSym2qa3}
\eeq
respectively. Now, employing the second and first symmetries of the inequality in Eqs. (\ref{preSym2qa1}) and (\ref{preSym2qa2}), respectively, we get the remaining anti-commutation relations
\beq
A_2A_6|\psi\ra&=&-A_6A_2|\psi\ra,  \\ \label{preSym2qa4}
A_3A_4|\psi\ra&=&-A_4A_3|\psi\ra. \label{preSym2qa5}
\eeq

Now, as in Ref. \cite{IMO+20} we define the subspace 
\beq \label{defInv}
    V &:=& \mathrm{span} \{ |\psi\ra , A_1 |\psi\ra, A_5 |\psi\ra, A_1A_5 |\psi\ra \},
\eeq
and prove the following fact for it.
\begin{lemma}\label{invsubV}
$V$ is an invariant subspace under the action of all the observables $A_i$ for $i\in [1,6]$.
\end{lemma}
\begin{proof}
Firstly, under the action of $A_1$, $V$ is trivially invariant. Next, using the fact that
$A^2_5=\I$ and $\{A_1 ,A_5\}\ket{\psi}=0$, it follows that $A_5 V$ belongs to $V$ up to the phase factor $-1$. Having established that $V$ is invariant under the action of $A_1$ and $A_5$, from the first symmetry \eqref{symmetry1} of the inequality, it follows that $A_2$ and $A_4$  also leave the subspace $V$ invariant.

Next, we proceed to see the action of $A_3$ on all the elements of $V$.  
Using the relations given by Eqs. (\ref{MaxQ_const1}) and (\ref{MaxQ_const4}), it follows that 
$A_3\ket{\psi}=A_1A_5\ket{\psi}$, $A_3A_1\ket{\psi}=A_5\ket{\psi}$, $A_3A_5\ket{\psi}=A_1\ket{\psi}$ and $A_3A_1A_5\ket{\psi}=\ket{\psi}$. 
Therefore, we have arrived at $A_3 V \subseteq V$.
Using the fourth symmetry in \eqref{symmetry4} of the inequality in these arguments used to establish invariance of $V$ under the action of $A_3$, it then follows 
that $A_6 V \subseteq V$.
\end{proof}

Due to Lemma \ref{invsubV}, it suffices for our purpose to identify the form of the state $|\psi\ra$ and the operators $A_i$  restricted to the subspace $V$. 
In fact, the whole Hilbert space splits as $\mathcal{H}=V\oplus V^{\perp}$, 
where $V^{\perp}$ is an orthocomplement of $V$ in $\mathcal{H}$. 
Then, the fact that $V$ is an invariant subspace of all
the observables $A_i$ means that they have the 
following block structure 
\begin{equation}\label{block}
    A_i=\hat{A}_i\oplus A_i',
\end{equation}
where $\hat{A}_i=PA_iP$ with $P:\mathcal{H}\to V$ being a projection onto $V$. Since $A_i'$  
act trivially on $V$, that is $A_i'V=0$, which means that the
observed correlations giving rise to the maximal violation of 
the inequality (\ref{TncIneq}) come solely from the subspace $V$, 
in what follows we can restrict our attention to the 
operators $\hat{A}_i$.

\begin{lemma}\label{lem_comm_inv}
Suppose the maximal quantum violation of the inequality (\ref{TncIneq}) is observed. 
Then, $[ \hat{A}_1,\hat{A}_2 ] =[ \hat{A}_1,\hat{A}_3 ] = [ \hat{A}_2,\hat{A}_3 ] = [ \hat{A}_4,\hat{A}_5 ]=
[ \hat{A}_4,\hat{A}_6 ]=[ \hat{A}_5,\hat{A}_6 ] = [ \hat{A}_1,\hat{A}_4 ] =[ \hat{A}_2,\hat{A}_5 ] =[ \hat{A}_3,\hat{A}_6 ]= 0$.
\end{lemma}

\begin{proof}
It is enough to show explicitly that $[ \hat{A}_1,\hat{A}_2 ] =[ \hat{A}_1,\hat{A}_4 ]=0$. Then, by using symmetries of the inequality (\ref{TncIneq}), the other commutators can be argued to vanish. 

To show that $[ \hat{A}_1,\hat{A}_2 ]=0$, we prove that  the action of the commutator $[A_1,A_2]$ on 
the invariant subspace vanishes. First, from the relations given by Eq. (\ref{MaxQ_const1}), we have
$A_1A_2\ket{\psi}=A_2A_1\ket{\psi}=A_3\ket{\psi}$ which implies that $[A_1,A_2]\ket{\psi}=0$. 
Second, it also follows from the relations given by Eq. (\ref{MaxQ_const1}), that $A_1A_2A_1\ket{\psi}=A_2A_1A_1\ket{\psi}=A_2\ket{\psi}$ which implies that $[A_1,A_2]A_1\ket{\psi}=0$.
Third, using the relation given by Eq. (\ref{MaxQ_const4}), we have $A_1A_2A_5\ket{\psi}=A_1\ket{\psi}$. On the other hand, using the relations given by Eqs. (\ref{MaxQ_const4}) and (\ref{MaxQ_const1}), we have  
$A_2A_1A_5\ket{\psi}=A_2A_1A_2\ket{\psi}=A_1\ket{\psi}$. Therefore, $[A_1,A_2]A_5\ket{\psi}=0$.
Finally, let us check whether  $[A_1,A_2]A_1A_5\ket{\psi}$ also vanishes. In the third case above,
we have shown that $A_2A_1A_5\ket{\psi}=A_1\ket{\psi}$. Using this relation,
it follows that $A_1A_2A_1A_5\ket{\psi}=\ket{\psi}$. On the other hand, using  
the relations given by Eq. (\ref{MaxQ_const4}) and (\ref{MaxQ_const1}), we have  $A_2A_1A_1A_5\ket{\psi}=\ket{\psi}$.
Therefore, we have arrived at $[A_1,A_2]A_1A_5\ket{\psi}=0$.

Using the symmetries \eqref{symmetry2}  and  \eqref{symmetry3} of the inequality in the above arguments, we also have $[A_2,A_3] V=0$ and $[A_1,A_3] V=0$, respectively. The fact that the action of the commutators $[A_1,A_2]$,  $[A_2,A_3]$ and $[A_1,A_3]$ vanish on $V$ implies that the action of the commutators $[A_4,A_5]$,  $[A_5,A_6]$ and $[A_4,A_6]$ on $V$ vanishes as  well. This follows from the fourth symmetry \eqref{symmetry4} of the inequality.

Next, we proceed to demonstrate that the action of the commutator $[A_1,A_4]$ on 
the invariant subspace vanishes. First, using Eq. (\ref{MaxQ_const3}), we have $A_1A_4\ket{\psi}=A_4A_1\ket{\psi}=\ket{\psi}$ which implies that $[A_1,A_4]\ket{\psi}=0$. Second, using the same relation,
we have $A_1A_4A_1\ket{\psi}=A_4\ket{\psi}$, on the other hand, $A_4A_1A_1\ket{\psi}=A_4\ket{\psi}$. From this, it follows that $[A_1,A_4]A_1\ket{\psi}=0$. Third, 
\be
A_1A_4A_5\ket{\psi}=A_1A_6\ket{\psi}
=-A_1A_3\ket{\psi}=-A_2\ket{\psi}, \nonumber
\ee
where we have used relations given by Eqs. (\ref{MaxQ_const2}), (\ref{MaxQ_const5}) and
(\ref{MaxQ_const1}), respectively. On the other hand,  we also have
\begin{align}
&&A_4A_1A_5=A_4A_1A_2\ket{\psi}=A_4A_3\ket{\psi}=-A_4A_6\ket{\psi} \nonumber \\
&&=-A_4A_4A_5\ket{\psi}=-A_2\ket{\psi} \nonumber
\end{align}
where in the first line we have used relations 
given by Eq. (\ref{MaxQ_const4}), (\ref{MaxQ_const1}) and (\ref{MaxQ_const5}), respectively,
and then we have employed  the relations given by Eqs. (\ref{MaxQ_const2}) and (\ref{MaxQ_const4}), respectively. Combining the above two equations imply that $[A_1,A_4]A_5\ket{\psi}=0$. Finally,
\begin{align}
&A_1A_4A_1A_5\ket{\psi}=A_1A_4A_1A_2\ket{\psi}=A_1A_4A_3\ket{\psi} \nonumber \\
&=-A_1A_4A_6\ket{\psi}=-A_1A_5\ket{\psi}=-A_1A_2\ket{\psi}=-A_3\ket{\psi} \nonumber
\end{align}
where in the first line we have used the relations given by Eqs. (\ref{MaxQ_const4}) and (\ref{MaxQ_const1}), respectively, and then we have employed the relations given by Eqs. (\ref{MaxQ_const5}), (\ref{MaxQ_const2}), (\ref{MaxQ_const4}) and (\ref{MaxQ_const1}), respectively.
On the other hand, we also have
\begin{equation}
A_4A_1A_1A_5\ket{\psi}=A_6\ket{\psi}=-A_3\ket{\psi}, \nonumber    
\end{equation}
where we have employed the relations given by Eqs. (\ref{MaxQ_const2}) and (\ref{MaxQ_const5}), respectively. Combining these two equations, we get that  $[A_1,A_4]A_1A_5\ket{\psi}=0$. 

Using the symmetries of the inequality in \eqref{symmetry1} and \eqref{symmetry2} in
the above arguments used to demonstrate that the action of the commutator   $[ A_1,A_4 ]$ vanishes on $V$, we will get $[ A_2,A_5 ] V =0$ and $[ A_3,A_6 ] V =0$, respectively.

\end{proof}

\begin{lemma}\label{lemmaanticomm}
Suppose the maximal quantum violation of the inequality (\ref{TncIneq}) is observed. 
Then, \beq\label{anticominv}
&&\{\hat A_1, \hat A_5\} = \{\hat A_1,\hat A_6\}  \nonumber \\
&=& \{\hat A_2,\hat A_4\}= \{\hat A_2,\hat A_6\}  \nonumber \\
&=& \{\hat A_3,\hat A_4\} = \{\hat A_3,\hat A_5\}= 0.
\eeq
\end{lemma}
\begin{proof}
It is enough to show explicitly that $\{ \hat{A}_1,\hat{A}_5 \} =0$, by demonstrating that the anti-commutator $\{A_1,A_5\}$ vanishes on the invariant subspace. First, in Eq. (\ref{preSym2q}), it has been
shown that $\{A_1,A_5\}\ket{\psi}=0$. Second, using the fact that $A_1A_5\ket{\psi}=-A_5A_1\ket{\psi}$, we have $A_1A_5A_1\ket{\psi}=-A_5\ket{\psi}$, on the other hand, $A_5A_1A_1\ket{\psi}=A_5\ket{\psi}$, which implies $\{A_1,A_5\} A_1 \ket{\psi}=0$.
Third, using  the fact that $A_1A_5\ket{\psi}=-A_5A_1\ket{\psi}$, we have $A_1A_5A_5\ket{\psi}=A_1\ket{\psi}$, on the other hand,
$A_5A_1A_5\ket{\psi}=-A_1\ket{\psi}$. Therefore, we get  $\{A_1,A_5\}A_5\ket{\psi}=0$. Finally, using the anticommutation relation $\{A_1,A_5\}\ket{\psi}=0$, we obtain $A_1A_5A_1A_5\ket{\psi}=-\ket{\psi}$ and $A_1A_5A_5A_1\ket{\psi}=\ket{\psi}$. Therefore, we have $\{A_1,A_5\}A_1A_5\ket{\psi}=0$. 

Using symmetries of the inequality (\ref{TncIneq}) in Eqs. \eqref{symmetry1}-\eqref{symmetry4}  as used to demonstrate Eq. (\ref{anticom}), it then follows that the other anti-commutators in Lemma. \ref{lemmaanticomm} can similarly be shown to vanish.
\end{proof}
%

Having proven Lemma. \ref{lemmaanticomm}, we can infer that the dimension $d$ of the subspace
$V$ is an even number by using the techniques as adopted in Refs. \cite{SJA22,Kan16,KST+19,SSK+19,BAS+20}. 
Thus, we can write the dimension $d=2k$ for some $k\in\mathbbm{N}$, and thus $V=\mathbbm{C}^2\otimes\mathbbm{C}^k$. Moreover, since $\dim V\leq 4$,
one concludes that $k=1,2$.

\begin{lemma}\label{lemma3}
Suppose the maximal quantum violation of the inequality (\ref{TncIneq}) is observed. Then, there exists a unitary $U= U_1 (\mathbbm{1}_2 \otimes U_2)$ acting on $V$ such that 
\beq  \label{Obs2Ctilde}
U \hat A_1 U^\dagger &=& X \otimes \mathbbm{1}, \quad U\hat A_2U^\dagger =  \mathbbm{1} \otimes Z,  \nonumber  \\ U\hat A_3U^\dagger &=&  \pm X \otimes Z, \quad
U\hat A_4U^\dagger =  \mathbbm{1} \otimes X ,\nonumber  \\ U\hat A_5U^\dagger &=&  Z \otimes \mathbbm{1},  \quad U\hat A_6U^\dagger = \pm Z \otimes X.
\eeq
\end{lemma}
\begin{proof}
First, from Lemma \ref{lemmaanticomm}, we have $\{\hat{A}_1,\hat A_5 \} = 0$ which implies that there exists a unitary $U_1$ acting on $V$ such that
\begin{eqnarray}\label{IdA1B1}
U_1^{\dagger}\, \hat{A}_1\, U_1 &=& X \otimes \mathbbm{1}_k, \\
U_1^{\dagger}\, \hat{A}_5\, U_1&=& Z \otimes \mathbbm{1}_k,
\end{eqnarray}
where, as already mentioned, the dimension $d$ of the subspace $V$ is given by $d=2k$ for some $k=1,2$. Using then the above form of $\hat{A}_1$ and $\hat{A}_5$ 
and the relations in Lemma \ref{lem_comm_inv} and \ref{lemmaanticomm} we can write the remaining operators as follows:
\beq
\label{form1}U^{\dag}_1\,\hat{A}_2\,U_1 &=& \mathbbm{1}_2 \otimes M, \\
U^{\dag}_1\,\hat{A}_3\,U_1 &=& X \otimes N, \\
U^{\dag}_1\,\hat{A}_4\,U_1 &=& \mathbbm{1}_2 \otimes O, \\
U^{\dag}_1\,\hat{A}_6\,U_1 &=& Z \otimes P, 
\eeq
where $M,N,O,P$ are Hermitian involutions acting on the subspace of dimension $k$. 
To show explicitly how the above equations are obtained let us focus on $\hat{A}_2$;
the proof for the other observable is basically the same. Since $\hat{A}_2$ acts on
$\mathbbm{C}^2\otimes\mathbbm{C}^k$, it can be decomposed in the Pauli basis as
\begin{equation}
    U_1^{\dagger}\,\hat{A}_2\,U_1=\mathbbm{1}_2\otimes M_1+X\otimes M_2+Y\otimes M_3+Z\otimes M_4,
\end{equation}
where $Y$ is the third Pauli matrix and $M_i$ are some Hermitian matrices acting on $\mathbbm{C}^k$. Now, it follows from the fact that $\hat{A}_2$ commutes with $\hat{A}_1$, that $M_3=M_4=0$. Then, from $[\hat{A}_2,\hat{A}_5]=0$, one obtains that
$M_2=0$, and, by putting $M_1=M$, we arrive at Eq. (\ref{form1}).

Second, from Lemma \eqref{lemmaanticomm}, we also have $\{ \hat{A}_2,\hat{A}_4\} = 0$ which is equivalent to $\{M,O\}=0$. Since both $M$ and $O$ are involutions, one concludes that $k=2$, or, equivalently, that $\mathbbm{C}^k=\mathbbm{C}^2$. In what follows, using $\{M,O\}=0$, we can fix $M=Z$ and $O=X$. Now, 
from Lemma \ref{lem_comm_inv} and \ref{lemmaanticomm}, using $[\hat{A}_2,\hat{A}_3] = 0$ and $\{\hat{A}_3,\hat{A}_4\}= 0$, we find that $[N,Z]=0$ as well as $\{N,X\}=0$. This is possible if $N=\pm Z$. Similarly, using $\{ \hat{A}_2,\hat{A}_6\} = 0$ and $[\hat{A}_4,\hat{A}_6] = 0$, we find that $O=\pm X$. 
Therefore, there exists another unitary transformation $U_2:\mathbbm{C}^k\to\mathbbm{C}^k$ such that
\begin{eqnarray*}
    U_2^{\dagger}\,M\, U_2&=& Z,\quad 
    U_2^{\dagger}\,N\, U_2=\pm Z,\\
    U_2^{\dagger}\,O\, U_2&=& X,\quad
    U_2^{\dagger}\,P\, U_2=\pm X.
\end{eqnarray*}
Then, we find the relations in Eq. \eqref{Obs2Ctilde}.
%
\end{proof}

We have thus arrived at the main result of this paper.

\begin{thm}\label{Theo2qubit}
If a quantum state $|\psi \ra$ and a set of measurements $A_i$  with $i\in [1,6]$ maximally violate the inequality \eqref{TncIneq}, then there exists a projection $P:\mathcal{H} \rightarrow V$ with $V=\mathbbm{C}^2 \otimes \mathbbm{C}^2 $ and a unitary $U$ acting on $V$ such that
\begin{align}\label{dupawolowa2q}
U^\dagger\, (P\hat A_1P^\dagger)\, U &= X \otimes \mathbbm{1}_2, \:\:  
U^\dagger\, (P\hat A_4P^\dagger) \, U =  \mathbbm{1}_2 \otimes X,  \nonumber \\
U^\dagger\, (P\hat A_2P^\dagger)\, U &=  \mathbbm{1}_2 \otimes Z, \:\:  
U^\dagger\, (P\hat A_5P^\dagger)\, U = Z \otimes \mathbbm{1}_2 ,  \nonumber \\  
U^\dagger\, (P\hat A_3P^\dagger)\, U &=  X \otimes Z,  \:\:  
U^\dagger\,(P\hat A_6P^\dagger)\,  U =  Z \otimes X ,
\end{align}
%
\vspace{-0.5cm}
\be
\text{and}\quad U (P|\psi \ra) = |\phi^+ \ra,
\ee
where 
$|\phi^+ \ra$ being the two-qubit maximally entangled state Eq. (\ref{2qmes}).
\end{thm}
\begin{proof}
A quantum state $\ket{\psi}$ that belongs to a Hilbert space $\mathcal{H}$ and a set of observables $A_i$ acting on $\mathcal{H}$ attain the maximal quantum violation of the inequality \eqref{TncIneq} if and only if they satisfy the set of Eqs. \eqref{MaxQ_const1}-\eqref{MaxQ_const5}. The algebraic relations induced by this set of equations let us prove Lemmas \ref{invsubV}-\ref{lemma3} which imply that there exists a projection $P:\mathcal{H} \rightarrow V \cong \mathbbm{C}^4 $ and a unitary $U = U_1 (\mathbbm{1}_2 \otimes U_2)$ acting on $V \cong \mathbbm{C}^4$ for which Eqs. (\ref{Obs2Ctilde})
hold true. Now, using stabilizer condition Eq. \eqref{MaxQ_const1}, we find that the sign in front of $U^\dagger \hat A_3 U$ can only be $+1$. Similarly, using Eq. \eqref{MaxQ_const2}, we can fix the sign of $U^\dagger \hat A_6 U$ to be $+1$.

From the above characterization of the observables, we can infer the form of the state $\ket{\psi}$. Indeed, after plugging  $A_i=\hat{A}_i \oplus A_i $ with $\hat{A}_i$ as given by Eq. \eqref{dupawolowa2q},  into the conditions (\ref{MaxQ_const3}) and (\ref{MaxQ_const4}) one realizes that the latter is simply the stabilizing conditions of the two-qubit maximally entangled state and thus $U (P\ket{\psi})=\ket{\phi^+}$. This completes the proof.
\end{proof}
Given this self-testing statement, we will now show in the following theorem that the above self-testing protocol is robust to very small errors ($\epsilon$) which might come from noise or experimental imperfections. A rigorous proof of the following fact has been furnished in Appendix \ref{robustnessproof}.
\begin{thm}\label{robust-main}
    Suppose a quantum state $\ket{\psi}$ and a set of measurements $A_i$ with $i\in [1,6]$ in $\mathcal{H}$, provide a non-maximal violation $5-\epsilon$ of the inequality \eqref{TncIneq}. Then, there exists a projection $P:\mathcal{H} \rightarrow V$ (with $V=\mathbbm{C}^2 \otimes \mathbbm{C}^2$) and a unitary $U$ acting on $V$ such that the measurements $\hat{A}_i$ and the state $\ket{\hat{\psi}}$ acting on $V$, satisfies the following relations,
    \begin{align}\label{robustness_statement}
        &\lVert U^{\dagger}\hat{A}_i^{noisy}U-\hat{A}_i^{opt}\rVert\leq m_1\sqrt{\epsilon}+m_2\epsilon^{1/4} \nonumber\\
        & |\braket{\hat{\psi}|\phi^{+}}|^2\geq 1-\left(s_1\epsilon+s_2\epsilon^{3/4}+s_3\sqrt{\epsilon}\right),
    \end{align}
    where $m_1,m_2,s_1,s_2,s_3$ are positive constants and $\hat{A}^{opt}$ are the optimal measurements obtained in Eq. \eqref{dupawolowa2q}.
\end{thm}
\section{Conclusions and Outlook} \label{Conc}
In order to demonstrate Kochen-Specker contextuality via non-contextuality inequalities, we need to assume certain compatibility relations between the measurements by defining "contexts". On the other hand, with temporal 
inequalities, we do not need to assume any compatibility conditions on the quantum measurements. In Ref. \cite{IMO+20}, a self-testing protocol based on observing measurement statistics that imply Kochen-Specker contextuality was proposed. This protocol certifies a two-qubit maximally entangled state and a subset of two-qubit Pauli measurements for realizing a topological quantum computer.
In this work, we proposed a self-testing protocol based on a temporal inequality that certifies this two-qubit entangled state and measurements. Thus, our scheme provides an alternative certification than the approach based on Kochen-Specker contextuality as it does not assume those compatibility conditions.

As a future direction, it would be interesting to find a self-testing scheme using temporal inequalities which will work for general measurements. In our work, we need to assume that measurements are projective if we drop the compatibility assumption. Note that Ref. \cite{DMS+22} assumes that measurement device produces maximally mixed state as well as same measurement performed many times produces a particular outcome with certainty always in order to drop compatibility assumptions. Alternative, it was suggested that one can assume a full rank input state to the measurement device instead of a maximally mixed state in order to remove the compatibility assumptions \cite{2023arXiv230701333S}. Therefore, we are lacking a universal self-testing scheme where such assumptions will be lifted. Also, it would be useful to find temporal inequalities with which we can reproduce self-testing statements for other existing non-contextuality inequalities. In another work, we extend our protocol to certify $n$-qubit quantum states and measurements with $n \ge 3$ \cite{nqubit_temp}.  

\textit{Acknowledgements.}  This work is supported by the Polish National Science Centre through the SONATA BIS project No. 2019/34/E/ST2/00369. SS acknowledges funding through PASIFIC program call 2 (Agreement No. PAN.BFB.S.BDN.460.022 with the Polish Academy of Sciences). This project has received funding from the European Union’s Horizon 2020 research
and innovation programme under the Marie Skłodowska-Curie grant agreement No 847639 and from the Ministry of Education and Science of Poland.

\bibliography{bibliography}
\onecolumngrid
\appendix
\section{Robustness Analysis} \label{robustnessproof}
In this section we show that our certification scheme is robust to the experimental errors and imperfections, i.e., we can certify the state and measurements within a threshold fidelity close to the ideal state and measurements, even when the inequality \eqref{TncIneq} is violated non-maximally.  Let us assume that a maximal violation of the inequality \eqref{TncIneq} is obtained with an $\epsilon$ error, i.e., a non-maximal violation of $5-\epsilon$ is observed. We will now demonstrate that for $\epsilon\rightarrow 0$, the quantum realization is very close to the optimal quantum realization that gives the maximal bound. For a non-maximal violation of $5-\epsilon$ it is implied that the sequential correlations satisfy the following bounds:
\begin{align}\label{seqcorrfidelity}
    &\braket{A_1A_2A_3}_{\pi} \geq 1-2\epsilon, \qquad
    \braket{A_2A_1A_3}_{\pi} \geq 1-2\epsilon, \nonumber \\
    &\braket{A_4A_5A_6}_{\pi} \geq 1-2\epsilon, \qquad 
    \braket{A_5A_4A_6}_{\pi} \geq 1-2\epsilon, \nonumber \\
    &\braket{A_1A_4}_{\pi} \geq 1-\epsilon, \quad 
    \braket{A_2A_5}_{\pi} \geq 1-\epsilon, 
    \quad -\braket{A_3A_6}_{\pi} \geq 1-\epsilon,
\end{align}
for some $\epsilon>0$. From these error-prone correlations we will do all the steps that went into the self-testing in the error-free scenario, i.e., we will deduce the commutation relations, anti-commutation relations, invariant subspace and the states and measurements within a small error bound. The sequential correlations in \eqref{seqcorrfidelity} further bound the quantum expectation values in the following way 
%
\begin{align}
    \braket{(A_1A_2A_3+A_3A_2A_1)} \geq 2(1-4\epsilon) \quad \forall \quad \text{permutations}, \label{expfid1} \\
    \braket{(A_4A_5A_6+A_6A_5A_4)} \geq 2(1-4\epsilon) \quad \forall \quad \text{permutations}. \label{expfid2}
\end{align}
It can be seen that these relations are an approximate version of the optimal relations in \eqref{MaxQ_const1}-\eqref{MaxQ_const5}. Equipped with these relations we now present the following lemma that will be useful for further proofs. 
\begin{lemma}
Suppose the relations in \eqref{seqcorrfidelity}-\eqref{expfid2} are satisfied for some $\epsilon>0$. Then the following bounds hold true
\begin{align}
    &\lVert(A_1-A_2A_3)\ket{\psi}\rVert\leq 4\sqrt{\epsilon} \quad \forall \quad \text{permutations}, \label{srcbound1} \\
    &\lVert(A_4-A_5A_6)\ket{\psi}\rVert\leq 4\sqrt{\epsilon} \quad \forall \quad \text{permutations}, \label{srcbound2}\\
    &\lVert(A_1-A_4)\ket{\psi}\rVert\leq 2\sqrt{\epsilon},\label{srcbound3}\\
    &\lVert(A_2-A_5)\ket{\psi}\rVert\leq 2\sqrt{\epsilon},\label{srcbound4}\\
    &\lVert(A_3+A_6)\ket{\psi}\rVert\leq 2\sqrt{\epsilon}.\label{srcbound5}
\end{align}
\end{lemma}
\begin{proof}
We note the following
\begin{align*}
    \lVert(A_1-A_2A_3)\ket{\psi}\rVert= \sqrt{2-\braket{(A_1A_2A_3+A_3A_2A_1)}}\leq 4\sqrt{\epsilon},
\end{align*}
where we simplify the norm by using the fact that $A_i$ are unitary measurements and then we substitute \eqref{expfid1} to get the inequality.  In the noise-free scenario, this reduces to $A_1\ket{\psi}=A_2A_3\ket{\psi}$. In a similar way, we can obtain other relations. 
\end{proof}
Using these relations one can obtain the following error bounds on the anti-commutators in Eq. \eqref{anticom}. 
\begin{align}\label{apprxanticom}
    \|\{A_i,A_j\}\ket{\psi}\|\leq 14\sqrt{\epsilon}.
\end{align}
The error bound can be obtained by following the proof of Eq. \eqref{anticom} and replacing each step of "equality" in the proof with a triangle inequality and then using the bounds in \eqref{srcbound1}-\eqref{srcbound5}. For instance in the anti-commutation proof of $\{A_1,A_5\}\ket{\psi}=0$ in Eq. \eqref{preSym2q}, we do the following 
\begin{align}\label{apprxpreSym2q}
    \|(A_1A_5+A_5A_1)\ket{\psi}\| &\leq \lVert(A_1A_5-A_1A_2)\ket{\psi}\rVert+\lVert(A_1A_2-A_3)\ket{\psi}\rVert\nonumber+\lVert(A_3+A_6)\ket{\psi}\rVert\nonumber \\&+\lVert(-A_6+A_5A_4)\ket{\psi}\rVert+\lVert(-A_5A_4+A_5A_1)\ket{\psi}\rVert  \leq 14 \sqrt{\epsilon},
\end{align}
where we have used a chain of triangle inequalities for the vector norm, the error bounds from \eqref{srcbound1}-\eqref{srcbound5} and the fact that the vector norm is unitarily invariant. In the same way, we can get error bounds on the other anti-commutators from Eq.\eqref{anticom}.
\textit{Action of $A_i$ on subspace $V$}.-- Next, we will show that the invariant subspace $V := \mathrm{span} \{|\psi\ra , A_1 |\psi\ra, A_5 |\psi\ra, A_1A_5 |\psi\ra \}$ considered in \eqref{defInv} is now approximately invariant under the action of operators $A_i$ for $i \in \{1,\ldots,6\}$. We find that the operator $A_1$ keeps the space $V$ invariant. For $A_5$, the non-trivial transformations are: 1. $A_5A_1\ket{\psi}\approx-A_1A_5\ket{\psi}+14\sqrt{\epsilon}\ket{\xi}$, and $A_5A_1A_5\ket{\psi}\approx -A_1\ket{\psi}+14\sqrt{\epsilon}\ket{\xi}$ from Eq. \eqref{apprxpreSym2q}. Therefore, $A_5$ keeps the vectors from $V$ approximately invariant. Note that vector $\ket{\xi}$ is arbitrary, moreover, we can assume w.l.o.g. that $\ket{\xi}$ lives in a subspace orthogonal to $V$. 

Similarly, one can show that other operators keep the subspace invariant approximately. For example, the operator $A_2$ acting on $V$ changes it's vectors to $A_2\ket{\psi}\approx A_5\ket{\psi}+2\sqrt{\epsilon}\ket{\xi}$ and $A_2A_5\ket{\psi}\approx A_2^2\ket{\psi}+2\sqrt{\epsilon}A_2\ket{\xi}$ using Eq. \eqref{srcbound4}. Now, the vector $A_2A_1\ket{\psi}\approx A_1A_2\ket{\psi}+8\sqrt{\epsilon}\ket{\xi}$ due to the relation that $\|A_1A_2-A_2A_1\|\leq \|A_1A_2-A_3\|+\|A_3-A_2A_1\|\leq 8\sqrt{\epsilon}$ (using \eqref{srcbound1}). Then, it implies using Eq. \eqref{srcbound4} that $A_2A_1\ket{\psi}\approx A_1A_5\ket{\psi}+2\sqrt{\epsilon}(1+4 A_1)\ket{\xi}$. The last vector $A_2A_1A_5\ket{\psi} \approx A_2A_1A_2\ket{\psi}+2\sqrt{\epsilon}A_2A_1\ket{\xi}$ first, then $A_2A_1A_2\ket{\psi}\approx A_2^2A_1\ket{\psi}+8\sqrt{\epsilon}A_2\ket{\xi}$. Therefore, $A_2$ keeps $V$ invariant approximately. We tabulate all such transformation in the following Table \ref{tab:my_label}.
\begin{table}[H]
 \centering
\begin{tabular}{ ||c||c|c|c|c| } 
 \hline
 $V$ & $|\psi\ra$ & $A_1 |\psi\ra$ & $ A_5 |\psi\ra$ & $A_1A_5 |\psi\ra$\\ 
 \hline
 $A_1 V$ & $A_1 |\psi\ra$ & $|\psi\ra$ & $A_1A_5 |\psi\ra$ & $ A_5 |\psi\ra$\\ 
 \hline
 $A_2 V$ & $A_5\ket{\psi}+2\sqrt{\epsilon}\ket{\xi}$ & $A_1A_5\ket{\psi}+2\sqrt{\epsilon}\Theta_1\ket{\xi}$ & $\ket{\psi}+2\sqrt{\epsilon}A_2\ket{\xi}$ & $A_1\ket{\psi}+2\sqrt{\epsilon}\Theta_2\ket{\xi}$\\ 
 \hline
 $A_3 V$ & $A_1A_5\ket{\psi}+2\sqrt{\epsilon}\chi\ket{\xi}$ & $A_5\ket{\psi}+6\sqrt{\epsilon}\ket{\xi}$ & $A_1\ket{\psi}+2\sqrt{\epsilon}(2+A_3)\ket{\xi}$ & $\ket{\psi}+2\sqrt{\epsilon}A_3\chi\ket{\xi}$\\ 
 \hline
 $A_4 V$ & $A_1\ket{\psi}+2\sqrt{\epsilon}\ket{\xi}$ & $\ket{\psi}+2\sqrt{\epsilon}A_4\ket{\xi}$ & $-A_1A_5\ket{\psi}+2\sqrt{\epsilon}\Theta_5\ket{\xi}$ & $-A_5\ket{\psi}+2\sqrt{\epsilon}\Theta_3\ket{\xi}$\\ 
 \hline
 $A_5 V$ & $A_5\ket{\psi}$ & $-A_1A_5\ket{\psi}+14\sqrt{\epsilon}\ket{\xi}$ & $\ket{\psi}$ & $-A_1\ket{\psi}+14\sqrt{\epsilon}\ket{\xi}$\\ 
 \hline
 $A_6 V$ & $-A_1A_5\ket{\psi}+2\sqrt{\epsilon}\Theta_4\ket{\xi}$ & $A_5\ket{\psi}+2\sqrt{\epsilon}(2+A_6)\ket{\xi}$ & $A_1\ket{\psi}+6\sqrt{\epsilon}\ket{\xi}$ & $-\ket{\psi}+2\sqrt{\epsilon}A_6\Theta_4\ket{\xi}$\\ 
 \hline
\end{tabular}
    \caption{\textcolor{blue}{$A_iV\to V+f(\epsilon)$}.-- The table describes how the vectors spanning $V$ transforms under the action of $A_i$ in non-ideal scenario. Here, we denote $\chi=2+A_1$, $\Theta_1=1+4 A_1$, $\Theta_2=A_2(1+\Theta_4)$, $\Theta_3=2+A_4\Theta_4$, $\Theta_4=1+\chi$, and $\Theta_5=3+\chi$.}
    \label{tab:my_label}
\end{table}
From Table \ref{tab:my_label}, we see that indeed action of $A_i$ keeps $V$ invariant approximately. However, notice that the vectors inside $V$ are not necessarily orthogonal. Therefore, it is not guaranteed that $A_i$ acting on a arbitrary state from $V$ will keep it approximately within $V$. To proceed, we need to show the following 
\begin{align}
    \max_{\ket{\psi}_1,\ket{\psi}_2\in V}\|A_i\ket{\psi}_1-\ket{\psi}_2\|\leq f(\epsilon),
\end{align}
where $f(\epsilon)$ is some well-behaved function of $\epsilon$. This will ensure that the action of $A_i$ on $V$ will be approximately invariant. For this purpose, we state and prove the following lemma. Importantly, we can remove maximization using the Gram-Schmidt orthogonalization procedure.
\begin{lemma}\label{apprxsubspace}
Under the action of operators $A_i$ for $i\in \{1,\ldots,6\}$, the subspace $V := \mathrm{span} \{|\psi\ra , A_1 |\psi\ra, A_5 |\psi\ra, A_1A_5 |\psi\ra \},$ changes as following
\begin{align}
    \forall \ket{\psi}_1 \in V, \quad \exists \ket{\psi}_k\in V, \quad s.t. \quad \lVert A_i\ket{\psi}_1-\ket{\psi}_k\rVert\leq 2C_{i}\sqrt{\epsilon} , 
\end{align}
where $C_{i}$ with $i \in [1,6]$ are 
real, finite, and positive constants; and $k\in [1,5]$.
\end{lemma}
\begin{proof}
%
Lets denote the vectors that spans the subspace $V$ as $\ket{e_1}=|\psi\ra$, $\ket{e_2}=A_1|\psi\ra$, $\ket{e_3}=A_5|\psi\ra$, and , $\ket{e_4}=A_1A_5|\psi\ra$. And define the Gram-matrix ($\Gamma$) with the elements $\Gamma_{mn}=\bra{e_m}e_n\rangle$. Then, by applying the Gram-Schmidt (GS) orthogonalization procedure one finds an orthogonal basis that spans the subspace $V$. We denote this basis as $\{\ket{\phi_m}|m=1,\ldots,4\}$, where $\ket{\phi_m}$ are expressed in terms of the initial vectors $\ket{e_m}$ as \cite{HornJohnson}
\begin{align}\label{GS-basis}
 \ket{\phi_m}=\sum_n \left[\Gamma^{-1/2}\right]_{nm} \ket{e_n}.
\end{align}
%
Now, we notice that the Gram-matrix $\Gamma$ is invariant when we have following transformation:
\begin{itemize}
    \item $\ket{e_1}\to X_2\ket{e_1}=A_1\ket{\psi}$, $\ket{e_2}\to X_2\ket{e_2}=\ket{\psi}$, $\ket{e_3}\to X_2\ket{e_3}=A_1A_5\ket{\psi}$, and , $\ket{e_4}\to X_2\ket{e_4}=A_5\ket{\psi}$.
    \item $\ket{e_1}\to X_3\ket{e_1}=A_5\ket{\psi}$, $\ket{e_2}\to X_3\ket{e_2}=A_1A_5\ket{\psi}$, $\ket{e_3}\to X_3\ket{e_3}=\ket{\psi}$, and , $\ket{e_4}\to X_3\ket{e_4}=A_1\ket{\psi}$.
    \item $\ket{e_1}\to X_4\ket{e_1}=A_1A_5\ket{\psi}$, $\ket{e_2}\to X_4\ket{e_2}=A_5\ket{\psi}$, $\ket{e_3}\to X_4\ket{e_3}=A_1\ket{\psi}$, and , $\ket{e_4}\to X_4\ket{e_4}=\ket{\psi}$,
\end{itemize}
where the symbolic operators, $X_2, X_3, X_4$, mimic the action of $A_1, A_5$ and $A_1A_5$ respectively. Therefore, we have another three GS basis (with same Gram-matrix $\Gamma$) which also spans $V$. It is evident from Table \ref{tab:my_label}, we only need these four bases, $\{X_k\ket{e_n}\}\in V$ for our analysis below. 
Now, any vector $\ket{\psi}_k\in V$ , $\forall k\in[1,4]$ can be expressed in therms of these orthogonal bases as 
\begin{align}\label{Eq-A14}
    \ket{\psi}_k=\sum_m\alpha_{m}X_{k}\ket{\phi_m}=\sum_m\alpha_{m}\sum_n\left[\Gamma^{-1/2}\right]_{nm}X_{k}\ket{e_n},
\end{align}
where we denote $X_1=\mathbbm{1}$ and $\alpha_{m}$ are complex coefficients satisfying $\sum_i|\alpha_{m}|^2=1$. 
Notice here that the action of $X_k$ on the set of vectors $\ket{e_n}\in V$ in the Eq. \eqref{Eq-A14} is listed above. 

Equipped with above relations, we are now ready to prove the lemma. First we prove that $A_1\ket{\psi}_1=\ket{\psi}_2$ below, 
\begin{align*}
    A_1\ket{\psi}_1=A_1\left(\sum_m\alpha_{m}\sum_n\left[\Gamma^{-1/2}\right]_{nm}\ket{e_n}\right)
    =\sum_m\alpha_{m}\sum_n\left[\Gamma^{-1/2}\right]_{nm}X_2\ket{e_n}=\ket{\psi}_2, 
\end{align*}
where we use Eq. \eqref{Eq-A14}. Now, consider the action of $A_2$ on $V$, then using Table \ref{tab:my_label}
\begin{align*}
     A_2\ket{\psi}_1&=\sum_m\alpha_{m}\sum_n\left[\Gamma^{-1/2}\right]_{nm}A_2\ket{e_n}\approx \sum_m\alpha_{m}\sum_n\left[\Gamma^{-1/2}\right]_{nm}\{X_3\ket{e_n}+2\sqrt{\epsilon}f_n(A_2)\ket{\xi}\},\\
     &=\ket{\psi}_3+2\sqrt{\epsilon} \sum_m\alpha_{m}\sum_n\left[\Gamma^{-1/2}\right]_{nm}f_n(A_2)\ket{\xi},
\end{align*}
where $f_1(A_2)=1, f_2(A_2)=1+4A_1, f_3(A_2)=A_2,f_4(A_2)=A_2(4+A_1)$ using Table \ref{tab:my_label} and we use Eq. \eqref{Eq-A14} in the second line. Note that in the above equation, we find that the Gram-matrix remains invariant also under the action of $\{A_i|i\in[2,6]\}$ as $\langle e_n\ket{\xi}=0$.
By rearranging the terms, we find the below relation  
\begin{align*}
    \|A_2\ket{\psi}_1-\ket{\psi}_3\|\leq 2\sqrt{\epsilon}\sum_m\Big\|\sum_n\left[\Gamma^{-1/2}\right]_{nm}f_n(A_2)\ket{\xi}\Big\|\leq 2\sqrt{\epsilon} C_{2},
\end{align*}
where we used $|\alpha_m|\leq 1$ and identified the real and positive constant $C_{2}$ as
\begin{align*}
    C_{2}=\sum_m\Big\|\sum_n\left[\Gamma^{-1/2}\right]_{nm}f_n(A_2)\ket{\xi}\Big\|.
\end{align*}
Notice that it is straight forward to find the expression of $C_{2}$ by opening the summation over $n$. In this way a lot of terms will be simplified by noticing that $A_i^2=\mathbbm{1}$.

Similarly, we can prove that all the operators $A_i$ will keep $V$ approximately invariant. In the end, we get the following relations for rest of the operators for their actions on $V$
\begin{align}\label{robustsubspacepreservation}
    &  \forall \ket{\psi}_1 \in V \quad \exists \ket{\psi}_4 \in V, \quad s.t. \quad \lVert A_3\ket{\psi}_1-\ket{\psi}_4\rVert\leq 2\sqrt{\epsilon} C_{3}, 
    \nonumber \\
    &  \forall \ket{\psi}_1 \in V \quad \exists \ket{\psi}_2 \in V, \quad s.t. \quad \lVert A_4\ket{\psi}_1-\ket{\psi}_2\rVert\leq 2\sqrt{\epsilon} C_{4}, 
    \nonumber \\
    &  \forall \ket{\psi}_1 \in V \quad \exists \ket{\psi}_3 \in V, \quad s.t. \quad \lVert A_5\ket{\psi}_1-\ket{\psi}_3\rVert\leq 2\sqrt{\epsilon} C_{5}, 
    \nonumber \\
    &  \forall \ket{\psi}_1 \in V \quad \exists \ket{\psi}_4 \in V, \quad s.t. \quad \lVert A_6\ket{\psi}_1-\ket{\psi}_4\rVert\leq 2\sqrt{\epsilon} C_{6}, 
\end{align}
where $C_i$ with $i\in \{3,..,6\}$, are real, finite, and positive constants.
Moreover, the constants $C_i$ can in general be estimated by the following relation,
\begin{align}
     C_{i}=\sum_m\Big\|\sum_n\left[\Gamma^{-1/2}\right]_{nm}f_n(A_i)\ket{\xi}\Big\|,
\end{align}
where the expressions of $f_n(A_i)$ can be found from the Table \ref{tab:my_label}. Hence, we prove the lemma.
\end{proof}

Due to Lemma \ref{apprxsubspace}, the operators $A_i$ are approximately restricted to the subspace $V$, i.e., we can write them as $A_i= \begin{pmatrix} \hat{A}_i & B^{\dagger}\\ B & A_i' \end{pmatrix}$, where similar to the error free scenario $\hat{A}_i=PA_iP$ with $P:\mathcal{H}\rightarrow V$ projects onto $V$, $A_i'$ acts trivially on $V$ and the correlation terms between $V$ and $V^{\perp}$, represented with $B$ are of the order of $2C_{i}\sqrt{\epsilon}$. Since the action of $A_i'$ on $V$ is trivial and $B$ on $V$ is of the order of $2C_{i}\sqrt{\epsilon}$, we can conclude that the contribution to the non-maximal violation comes mostly from the operators $\hat{A}_i$. For the remaining proof, we can limit ourselves to $\hat{A}_i$ only. 
Furthermore, due to \eqref{robustsubspacepreservation} the all the 
 eigenvalues $a_i$ of operators $A_i$ respectively, are bounded as following 
\begin{align}\label{robusteigen}
    |a_1|=1, \quad
    1-2C_{i}\sqrt{\epsilon}\leq|a_i|\leq 1.
\end{align}
Note that since $A_1$ keeps $V$ invariant, we have the unitary condition $\hat{A}_1^2=\mathbbm{1}$, whereas all the remaining hatted operators are not unitary anymore.
For robust self-testing, we also need to show that the commutators and anti-commutators from Lemma \ref{lem_comm_inv} and \ref{lemmaanticomm}, hold within a small error bound proportional to $\epsilon$. We derive these bounds in the next two lemmas. The derivations of the error bounds on the commutators and anti-commutators are obtained by following the proofs in the error-free scenario and replacing the "equalities" with the relevant triangle inequalities.
\begin{lemma}\label{apprxcommlemma}
Suppose a non-maximal quantum violation $5-\epsilon$ of the inequality \eqref{TncIneq} is observed. Then the following error bounds on the commutation relations can be obtained. 
\begin{align*}
    &\{\lVert[\hat{A}_1,\hat{A}_2]\rVert,\:\: \lVert[\hat{A}_1,\hat{A}_3]\rVert, \:\:  \lVert[\hat{A}_2,\hat{A}_3]\rVert\}\leq 4\left(C_{2}+C_{3}\right)\sqrt{\epsilon},\\ 
    &\{\lVert[\hat{A}_4,\hat{A}_5]\rVert,\:\:\lVert[\hat{A}_4,\hat{A}_6]\rVert,\:\: \lVert[\hat{A}_5,\hat{A}_6]\rVert\}\leq 4\left(C_{4}+C_{5}+C_{6}\right)\sqrt{\epsilon},\\
    &\lVert[\hat{A}_1,\hat{A}_4]\leq 4C_{4}\sqrt{\epsilon},\:\:\rVert\lVert[\hat{A}_2,\hat{A}_5]\rVert\leq 4(C_{2}+C_{5})\sqrt{\epsilon},\:\:\lVert[\hat{A}_3,\hat{A}_6]\rVert\leq 4(C_{3}+C_{6})\sqrt{\epsilon}.
\end{align*}
\end{lemma}
\begin{proof}
We will show the proof for error bound on $[\hat{A}_1,\hat{A}_2]$, and the remaining bounds are obtained in a similar way. To get a bound on $\lVert[\hat{A}_1,\hat{A}_2]\rVert$, we will prove that $\lVert[A_1,A_2]\rVert$ is bounded by $4(C_{2}+C_{3})\sqrt{\epsilon}$ in arbitrary vectors from subspace $V$.
Consider the vector $\ket{\psi}_1\in V$, we have 
\begin{align*}
    \lVert[A_1,A_2]\ket{\psi}_1\rVert \leq \lVert(A_1A_2-A_3)\ket{\psi}_1\rVert+\lVert(A_3-A_2A_1)\ket{\psi}_1\rVert, 
\end{align*}
where we use the triangle inequality.
Now, consider the term $\sum_i\lVert(A_1A_2-A_3)\ket{\psi}_1\rVert$ and evaluate it,
\begin{align*}
    \lVert(A_1A_2-A_3)\ket{\psi}_1\rVert\leq \lVert A_1A_2\ket{\psi}_1-\ket{\psi}_4\rVert+\lVert\ket{\psi}_4-A_3\ket{\psi}_1\rVert,
\end{align*}
where $\ket{\psi}_4 \in V$. According to Lemma \ref{apprxsubspace}, $\lVert \ket{\psi}_4-A_3\ket{\psi}_1\rVert\leq 2\sqrt{\epsilon}C_{3}$. We need to evaluate the first term above. Further notice that 
\begin{align*}
   \lVert A_1A_2\ket{\psi}_1-\ket{\psi}_4\rVert\leq \lVert A_1A_2\ket{\psi}_1-A_1\ket{\psi}_3\rVert+\lVert A_1\ket{\psi}_3-\ket{\psi}_4\rVert,
\end{align*} 
where $\ket{\psi}_3\in V$. 
Using the Lemma \ref{apprxsubspace}, we conclude that $\lVert A_1A_2\ket{\psi}_1-A_1\ket{\psi}_3\rVert \leq 2\sqrt{\epsilon}C_{2}$. Note that the relation $A_1\ket{\psi}_3=\ket{\psi}_4$ holds trivially. Therefore, we reach to the following expression after collecting all such contributions,
\begin{align*}
    \lVert[A_1,A_2]\ket{\psi}_1\rVert \leq 4 \sqrt{\epsilon}\left(C_{2}+C_{3}\right).
\end{align*}
Similarly, we show for all the commutators. This proves our claim.
\end{proof}
Next, we find the error bound for the anti-commutation relations in the following lemma.
\begin{lemma}\label{apprxanticommlemma}
Suppose a non-maximal quantum violation $5-\epsilon$ of the inequality \eqref{TncIneq} is observed. Then the following error bounds on the anti-commutation relations can be obtained.
\begin{align*}
    \{\lVert\{\hat{A}_1,\hat{A}_5\}\rVert,\:\:\lVert\{\hat{A}_2,\hat{A}_4\}\rVert,\:\:\lVert\{\hat{A}_3,\hat{A}_4\}\rVert,\:\: \lVert\{\hat{A}_3,\hat{A}_5\}\rVert
    ,\:\:\lVert\{\hat{A}_1,\hat{A}_6\}\rVert,\:\:\lVert\{\hat{A}_2,\hat{A}_6\}\rVert\}\leq 4C_{\gamma}\sqrt{\epsilon},
\end{align*}
where $C_\gamma=\sum_{i=2}^6C_i$ is also a finite constant.
\end{lemma}
\begin{proof}
We consider $\lVert\{\hat{A}_3,\hat{A}_4\}\rVert$, and prove the upper bound for it by showing that $\lVert\{A_3,A_4\}\rVert\leq 4C_{\gamma}\sqrt{\epsilon}$ for every element of $V$. 
Similar to the proof of Lemma \ref{apprxcommlemma}, consider any arbitrary $\ket{\psi}_1\in V$, then 
\begin{align*}
    \lVert\{A_3,A_4\}\ket{\psi}_1\rVert \leq \lVert
    (A_3A_4&-A_3A_1)\ket{\psi}_1\rVert+\lVert (A_3A_1-A_2)\ket{\psi}_1\rVert+\lVert (A_2-A_5)\ket{\psi}_1\rVert\\&+\lVert (A_5-A_4A_6)\ket{\psi}_1\rVert+\lVert(A_4A_6+A_4A_3)\ket{\psi}_1\rVert,
\end{align*}
where we use the triangle inequality. Then the Lemma \ref{apprxsubspace}, the first three terms in RHS are easy to compute, and, directly, $\lVert(A_3A_4-A_3A_1)\ket{\psi}_1\rVert\leq 2\sqrt{\epsilon}C_4$, $\lVert (A_3A_1-A_2)\ket{\psi}_1\rVert\leq 2\sqrt{\epsilon}(C_2+C_3)$ and $\lVert (A_2-A_5)\ket{\psi}_1\rVert\leq 2\sqrt{\epsilon}(C_2+C_5)$. The fourth one, can be estimated as,
$\lVert (A_5-A_4A_6)\ket{\psi}_1\rVert\leq 2\sqrt{\epsilon}(C_4+C_5+C_6)$ and the last one, can be expressed as,
$\lVert(A_4A_6+A_4A_3)\ket{\psi}_1\rVert\leq 2\sqrt{\epsilon}(C_6+C_3)$. Collecting all these contributions, we find that 
\begin{align*}
    \lVert\{A_3,A_4\}\ket{\psi}_1\rVert\leq 4C_{\gamma}\sqrt{\epsilon},
\end{align*}
where $C_{\gamma}=\sum_{i=2}^6C_i$.
In the same way, we can show the other anti-commutator bounds by the Lemma \ref{apprxsubspace}.
\end{proof}
Having proved Lemma \ref{apprxanticommlemma}, we can show that for a basis $\{\ket{\phi_i}\}$, we have $\Tr(\hat{A}_5\hat{A}_1\hat{A}_5+\hat{A}_1)=\sum_{i}\bra{\phi_i}\hat{A}_5\hat{A}_1\hat{A}_5+\hat{A}_1\ket{\phi_i}\leq \sum_{i}\lVert\bra{\phi_i}\{\hat{A}_1,\hat{A}_5\}\ket{\phi_i}\rVert\leq 4 C_{\gamma}\sqrt{\epsilon}$, where the first inequality is obtained using the triangle inequality and the next one by using the bound on $\lVert \{\hat{A}_1,\hat{A}_5\}\ket{\psi}_1\rVert$. Therefore, we can deduce that $\Tr(\hat{A}_1)\leq 4C_{\gamma}\sqrt{\epsilon}$ which is possible only when $\Tr(\hat{A}_1)=0$. Similarly for all other operators $\hat{A}_i$. As in the noise-free case, we can write the dimension $d=2k$ for $k\in \mathbb{N}$, and thus $V=\mathbb{C}^2\otimes \mathbb{C}^k$. Also since, dim$V \leq 4$, we can only have $k=1,2$. Now we are ready to certify the measurements in the following lemma
\begin{lemma} \label{robustoperatorslemma}
Suppose a non-maximal quantum violation $5-\epsilon$ of the inequality \eqref{TncIneq} is observed. Then, there exists a basis in $V$ such that
\begin{align}  \label{apprxObs2Ctilde}
&\hat{A}_1=X\otimes \mathbbm{1}, \quad\quad\quad\quad\quad\quad\quad\quad\quad\quad\quad\quad\:\:
\lVert \hat{A}_2-\mathbb{1}\otimes Z\rVert \leq 2c_2\sqrt{\epsilon}+8C\epsilon^{1/4},\nonumber\\
&\lVert \hat{A}_3-X\otimes Z\rVert \leq \nonumber 2c_3\sqrt{\epsilon}+16C\epsilon^{1/4}, \:\:\:\quad\quad \lVert \hat{A}_4-\mathbb{1}\otimes X\rVert \leq  2c_4\sqrt{\epsilon}+16C\epsilon^{1/4}, \nonumber\\ 
&\lVert A_5-Z\otimes \mathbb{1}\rVert\leq  6C_{\gamma}\sqrt{\epsilon}+2C\epsilon^{\frac{1}{4}}, \quad {\rm and} \quad \lVert\hat{A}_6-Z\otimes X\rVert \leq    2c_6\sqrt{\epsilon}+16C\epsilon^{1/4},
\end{align}
where $c_2=C_{2}+2\bar{C}_2$, $c_3=C_3+2\bar{C}_3+2\bar{C}_4+2C_{\gamma}-2C_2$, $c_4=C_2+2\bar{C}_2+2\bar{C}_4+C_{\gamma}$ and $c_6=C_{\gamma}+C_2+C_3+2\bar{C}_2+2\bar{C}_6$.
\end{lemma}
\begin{proof}
From Lemma \ref{apprxanticommlemma}, we have $\lVert \{A_1,A_5\}\rVert\leq 4C_{\gamma}\sqrt{\epsilon}$, which implies that there exists a unitary $\mathcal{U}_1$ acting on $V$ such that 
\begin{align}\label{a1a5}        
\hat{A}_1^{(1)}=X\otimes \mathbb{1}, \quad {\rm and} \quad
\hat{A}_5^{(1)}=\mathbb{1}\otimes M_1+X\otimes M_x+Y\otimes M_y+Z\otimes M_z,
\end{align}
where action of $\mathcal{U}_k$ on the operator $\mathcal{O}$ is denoted as $\mathcal{U}_k^{\dagger}\mathcal{O}\mathcal{U}_k=\mathcal{O}^{(k)}$ for our convenience. 
Substituting the above form in the bound $\lVert\{\hat{A_1},\hat{A_5}\}\rVert\leq 4C_{\gamma}\sqrt{\epsilon}$, we get that $\lVert \{\hat{A}_1^{(1)},\hat{A}_5^{(1)}\}\rVert = \lVert \{X\otimes \mathbb{1},\hat{A}_5^{(1)}\}\rVert \leq 4C_{\gamma}\sqrt{\epsilon}$, which further implies that
\begin{align}\label{A5bound}
 \lVert\mathbb{1}\otimes M_1+X\otimes M_x\rVert= \lVert \hat{A}_5^{(1)}-Y\otimes M_y-Z\otimes M_z\rVert \leq 2C_{\gamma}\sqrt{\epsilon}.
\end{align}
 Therefore, it holds that $\lVert (Y\otimes M_y+Z\otimes M_z)^2\rVert\geq (1-2C_{\gamma}\sqrt{\epsilon})^2$. From this approximate unitary, we can deduce the following 
\begin{align}\label{apprxMyMz}
&\lVert\mathbb{1}\otimes(M_y^2+M_z^2)+YZ\otimes[M_y,M_z]\rVert\geq (1-2C_{\gamma}\sqrt{\epsilon})^2,\nonumber\\
&\implies
   (1-2C_{\gamma}\sqrt{\epsilon})^2 \leq \lVert M_y^2+M_z^2\rVert\leq 1, \quad\text{and}\quad \lVert[M_y,M_z]\rVert\leq 4C_{\gamma}\sqrt{\epsilon}-16C_{\gamma}\epsilon\leq 4C_{\gamma}\sqrt{\epsilon}
\end{align}
It is known that for such almost commuting Hermitian matrices we can find another pair of commuting Hermitian matrices $M_y',M_z'$ arbitrarily close to $M_y,M_z$ respectively \cite{kachkovskiy2016distance,2009CMaPh.291..321H}, which in our case
can be stated as
 \begin{align}\label{apprxcommtheorem}
 [M_y',M_z']=0, \quad \text{s.t.,} \quad\{\lVert M_y-M_y'\rVert,\:\:\lVert M_z-M_z'\rVert\}\leq C\epsilon^{1/4},
 \end{align}
where $C$ is a constant independent of $M_y,M_z$, however, it is function of $\epsilon$ and grows slower than any power of $\epsilon$. From \eqref{apprxMyMz} it can also be seen that $\lVert M_y-\cos\alpha\mathbb{1}\rVert\leq 2C_{\gamma}\sqrt{\epsilon}$ and $\lVert M_z-\sin\alpha\mathbb{1}\rVert\leq 2C_{\gamma}\sqrt{\epsilon}$, $\forall \alpha\in[0,2\pi]$, which further implies the following 
\begin{align*}
\lVert M_y'-\cos\alpha\mathbbm{1}\rVert \leq  \lVert M_y'-M_y\rVert+\lVert M_y-\cos\alpha\mathbbm{1}\rVert\leq 2C_{\gamma}\sqrt{\epsilon}+C\epsilon^{1/4}.    
\end{align*} 
And similarly, we have $\lVert M_z'-\sin\alpha\mathbbm{1}\rVert \leq 2C_{\gamma}\sqrt{\epsilon}+C\epsilon^{1/4}$.
Now by substituting these bounds in \eqref{A5bound} we get the following by using triangle inequalities
\begin{align}
    \lVert \hat{A}_5^{(1)}-\cos\alpha & Y\otimes \mathbb{1}-\sin\alpha Z\otimes \mathbb{1}\rVert - \lVert Y\otimes(M_y'-\cos\alpha\mathbb{1})\rVert-\lVert Z\otimes(M_z'-\sin\alpha\mathbb{1})\rVert \leq  2C_{\gamma}\sqrt{\epsilon},\nonumber\\
   & \implies
    \lVert \hat{A}_5^{(1)}-\cos\alpha Y\otimes \mathbb{1}-\sin\alpha Z\otimes \mathbb{1}\rVert\leq 6C_{\gamma}\sqrt{\epsilon}+2C\epsilon^{1/4}.
\end{align}
Therefore, we can proceed further by applying another unitary $\mathcal{U}_2$ such that $\mathcal{U}_2^{\dagger}( \cos\alpha Y\otimes \mathbb{1}+\sin\alpha Z\otimes \mathbb{1})\mathcal{U}_2\to Z\otimes\mathbb{1}$ 
and keeps $\hat{A}_1^{(1)}=X\otimes \mathbb{1}$ unchanged as the rotation is in the $y-z$ plane. Therefore we have a unitary $U_1=\mathcal{U}_1\mathcal{U}_2$ such that
\begin{align}\label{A1A5U1}
    U_1^{\dagger}\hat{A}_1U_1=X\otimes \mathbb{1},\quad
    \lVert U_1^{\dagger}\hat{A}_5U_1-Z\otimes \mathbb{1}\rVert\leq 6C_{\gamma}\sqrt{\epsilon}+2C\epsilon^{1/4}.
\end{align}
Now that we have the above form of $\hat{A}_1$ and $\hat{A}_5$, we can write the remaining operators in the same basis as following by using the Lemmas \ref{apprxcommlemma} and \ref{apprxanticommlemma}
\begin{align}\label{afterfirstunitary}
    &\lVert U_1^{\dagger}\hat{A}_2U_1-\mathbb{1}\otimes N\rVert \leq 2\bar{C}_2\sqrt{\epsilon}+4C\epsilon^{1/4}, \quad \quad \lVert U_1^{\dagger}\hat{A}_3U_1-X\otimes O\rVert \leq 2\bar{C}_{3}\sqrt{\epsilon}+4C\epsilon^{1/4},\nonumber  \\
    &\lVert U_1^{\dagger}\hat{A}_4U_1-\mathbb{1}\otimes P\rVert \leq 2\bar{C}_{4}\sqrt{\epsilon}+4C\epsilon^{1/4},\quad\quad\lVert U_1^{\dagger}\hat{A}_6U_1-Z\otimes Q\rVert \leq 2\bar{C}_{6}\sqrt{\epsilon}+4C\epsilon^{1/4}, 
\end{align}
where $\bar{C}_2=(4+\sqrt{2})C_{2}+\sqrt{2}C_3+4C_5$, $\bar{C}_{3}=4C_\gamma+\sqrt{2}(C_2+C_3)$, $\bar{C}_{4}=(4+\sqrt{2})C_4+4(C_5+C_6)$, $\bar{C}_{6}=\sqrt{C_{\gamma}^2+C_6^2}+4(C_4+C_5+C_6)$, and $N,O,P,Q$ are hermitian operators acting on the subspace of dimension $k\leq 2$. We will show explicitly how to obtain the relation for $\hat{A}_2$ and the remaining operators can be obtained similarly. First, by expanding $\hat{A}_2$ in the Pauli basis as following
\begin{align}\label{A2paulibasis}
    \hat{A}_2=\mathbb{1}\otimes N+X\otimes N_x+Y\otimes N_y+Z\otimes N_z.
\end{align}
Substituting this form in the bounds on commutation relations of $\hat{A}_2$ with $\hat{A}_1$ and $\hat{A}_5$ from Lemma \ref{apprxcommlemma}, we will get the following relations respectively
\begin{align}\label{A2bound}
    &\lVert Y\otimes N_y+Z\otimes N_z \rVert=\lVert\hat{A}_2 -\mathbb{1}\otimes N-X\otimes N_x\rVert\leq 2(C_{2}+C_3)\sqrt{\epsilon},\nonumber \\
       & \lVert Y\otimes N_y+X\otimes N_x \rVert=\lVert\hat{A}_2 -\mathbb{1}\otimes N-Z\otimes N_z\rVert\leq 8(C_{2}+C_5)\sqrt{\epsilon}+4C\epsilon^{1/4}.
\end{align}
Now notice that since $\hat{A}_2$ is almost a unitary matrix, the operators $N_x,N_y,N_z$ approximately commute with each other. We can further show by manipulating the above expression that
\begin{align*}
    \lVert (Y\otimes N_y+&Z\otimes N_z)^2 \rVert=\lVert \mathbb{1}\otimes N_y^2+\mathbb{1}\otimes N_z^2+YZ[\mathbbm{1}\otimes N_y,\mathbbm{1}\otimes N_z] \rVert\leq 4(C_{2}+C_3)^2\epsilon \\
   &\implies \lVert \mathbb{1}\otimes N_y^2+\mathbb{1}\otimes N_z^2\rVert\leq 4(C_{2}+C_3)^2\epsilon+\lVert YZ[\mathbbm{1}\otimes N_y,\mathbbm{1}\otimes N_z] \rVert\leq 8(C_{2}+C_3)^2\epsilon\\ 
    &\implies\lVert N_y \rVert,\lVert N_z \rVert \leq (C_{2}+C_3) \sqrt{8\epsilon},
\end{align*}
where we have expanded using the triangle inequality and then upper bounded $\lVert YZ[\mathbbm{1}\otimes N_y,\mathbbm{1}\otimes N_z]\rVert$ by $(2(C_{2}+C_3)\sqrt{\epsilon})^2$, using the fact that  $1-2C_{2}\sqrt{\epsilon}\leq |a_2|\leq 1$ from $\eqref{robusteigen}$. 
Substituting these in the second relation of \eqref{A2bound}, we will get the bound $\lVert \hat{A}_2-\mathbb{1}\otimes N\rVert\leq 2\bar{C}_2\sqrt{\epsilon}+4C\epsilon^{1/4}$, with $\bar{C}_2=(4+\sqrt{2})C_{2}+\sqrt{2}C_3+4C_5$. In a similar way we can get the relations for the remaining operators in $\hat{A}_3,\hat{A_4},\hat{A}_6$ in \eqref{afterfirstunitary}.

Now, using the bounds on $\hat{A}_2$ and $\hat{A}_4$ \eqref{afterfirstunitary} and the bound on $\{\hat{A_2},\hat{A}_4\}$ from Lemma \ref{apprxanticommlemma} we can show that
\begin{align}\label{NPanticomm}
    \lVert\{\mathbb{1}\otimes N,\mathbb{1}\otimes P\}\rVert\leq \lVert (\mathbbm{1}\otimes N-&A_2)\mathbbm{1}\otimes P\rVert+\lVert A_2(\mathbb{1}\otimes P-A_4)\rVert+\lVert \mathbbm{1}\otimes P(\mathbbm{1}\otimes N-A_2)\rVert+\lVert (\mathbb{1}\otimes P-A_4)A_2\rVert\nonumber \\ & +\lVert \{A_2,A_4\}\rVert \leq 4(C_{\gamma}+\bar{C}_2+\bar{C}_4)\sqrt{\epsilon}+16C\epsilon^{1/4}.
\end{align}
We can now find a unitary operator $\mathcal{U}_3$ such that 
\begin{align}\label{NU3}
    \lVert \mathbb{1}\otimes N^{(3)}-\mathbb{1}\otimes Z \rVert \leq \lVert \mathbb{1}\otimes N -\hat{A}_2 \rVert + \lVert \hat{A}_2^{(3)} -\mathbb{1}\otimes Z \rVert\leq 2(C_{2}+\bar{C}_2)\sqrt{\epsilon}+4C\epsilon^{1/4}, 
\end{align}
where we have first used the triangle inequality and then the bounds on $\lVert\mathbbm{1}\otimes N-\hat{A}_2\rVert$ from \eqref{afterfirstunitary}, $\lVert \hat{A}_2-\mathbbm{1}\otimes Z\rVert$ from the fact that the minimum eigenvalue of $\hat{A}_2$ is $1-2C_{2}\sqrt{\epsilon}$. Now using the anti-commutation bound in \eqref{NPanticomm} and  the inequality in \eqref{NU3}, we can get the following relation
\begin{align*}
    \lVert\{\mathbbm{1}\otimes P^{(3)},\mathbbm{1}\otimes Z\} \rVert&\leq\lVert \{\mathbbm{1}\otimes P,\mathbbm{1}\otimes N\}\rVert+\lVert \{\mathbbm{1}\otimes P^{(3)},(\mathbbm{1}\otimes Z-\mathbbm{1}\otimes N^{(3)})\}\rVert \nonumber \\ &\leq 4(C_{\gamma}+C_2+2\bar{C}_2+\bar{C}_4)\sqrt{\epsilon}+24C\epsilon^{1/4},
\end{align*}
where we first use the triangle inequality and then apply the upper bounds. From the above relation, the form of $\mathbbm{1}\otimes P^{(3)}=p_1\mathbbm{1}\otimes \mathbbm{1}+p_x\mathbbm{1}\otimes X+p_y\mathbbm{1}\otimes Y+p_z\mathbbm{1}\otimes Z$ can be written as 
\begin{align*}
    \lVert \mathbbm{1}\otimes P^{(3)}-p_x\mathbbm{1}\otimes X-p_y\mathbbm{1}\otimes Y\rVert \leq 2(C_{\gamma}+C_2+2\bar{C}_2+\bar{C}_4)\sqrt{\epsilon}+12C\epsilon^{1/4}.
\end{align*}
By finding another unitary $\mathcal{U}_4$ which rotates the operator $p_xX+p_yY$ to $X$ and leaves $\mathbbm{1}\otimes Z$ invariant as it rotates the operators only in the $x-y$ plane. Thus, we can get the following by defining $U_2=\mathcal{U}_4\mathcal{U}_3$, 
\begin{align}\label{NPU4}     
&\lVert U_2^{\dagger}\mathbbm{1}\otimes N U_2-\mathbbm{1}\otimes Z\rVert\leq  2(C_{2}+\bar{C}_2)\sqrt{\epsilon}+4C\epsilon^{1/4}, \nonumber \\
&\lVert U_2^{\dagger}\mathbbm{1}\otimes P U_2-\mathbbm{1}\otimes X\rVert\leq 2(C_{\gamma}+C_2+2\bar{C}_2+\bar{C}_4)\sqrt{\epsilon}+12C\epsilon^{1/4}.
\end{align}
Thus by defining the effect of these unitaries as $U=U_2U_1$, and then substituting the relations of \eqref{NPU4} further in \eqref{afterfirstunitary} for $\hat{A}_2$ and $\hat{A}_4$, we will get  
\begin{align}\label{A2A4U2}
     &\lVert U^{\dagger}\hat{A}_4U-\mathbb{1}\otimes X\rVert \leq \lVert U^{\dagger}\hat{A}_4U-U_2^{\dagger}\mathbb{1}\otimes PU_2\rVert +\lVert U_2^{\dagger}\mathbb{1}\otimes PU_2-\mathbbm{1}\otimes X\rVert
     \leq   2(C_2+2\bar{C}_2+2\bar{C}_4+C_{\gamma})\sqrt{\epsilon}+16C\epsilon^{1/4}, \nonumber \\
      &\lVert U^{\dagger}\hat{A}_2U-\mathbb{1}\otimes Z\rVert \leq \lVert U^{\dagger}\hat{A}_2U-U_2^{\dagger}\mathbb{1}\otimes NU_2\rVert +\lVert U_2^{\dagger}\mathbb{1}\otimes NU_2-\mathbbm{1}\otimes Z\rVert
      \leq 2(C_{2}+2\bar{C}_2)\sqrt{\epsilon}+8C\epsilon^{1/4}.
\end{align}
Now to obtain the operators $\hat{A}_3,\hat{A}_6$ we need to determine $O$ and $Q$. To do this, we use the bounds on the commutator and anti-commutator of $\hat{A}_3(\hat{A}_6)$ with $\hat{A}_2$ and $\hat{A}_4$ from Lemmas \ref{apprxcommlemma} and \ref{apprxanticommlemma}. Let us note the following relations for $O$, 
\begin{align*}
    \lVert [X\otimes O,\mathbb{1}\otimes Z]\rVert \leq \lVert (X\otimes O-&\hat{A}_3)\mathbb \otimes Z\rVert+ \lVert \hat{A}_3(\mathbb{1}\otimes Z-\hat{A}_2)\rVert +\lVert (\hat{A}_2-\mathbb{1}\otimes Z)\hat{A}_3\rVert +\lVert \mathbb{1} \otimes Z(\hat{A}_3-X\otimes O)\rVert \\&+\lVert[\hat{A}_2,\hat{A}_3]\rVert  \leq 4(C_2+\bar{C}_2+C_3+\bar{C}_3)\sqrt{\epsilon}+24C\epsilon^{1/4}.\\
     \lVert \{X\otimes O,\mathbb{1}\otimes X\}\rVert\leq \lVert (X\otimes O-&\hat{A}_3)\mathbb{1}\otimes X\rVert+\lVert A_3(\mathbb{1}\otimes X-A_4)\rVert+\lVert \mathbb{1}\otimes X(X\otimes O-\hat{A}_3)\rVert+\lVert (\mathbb{1}\otimes X-A_4)A_3\rVert\\&+ \lVert\{A_3,A_4\}\rVert  \leq  4(C_{\gamma}+\bar{C}_3+\bar{C}_4)\sqrt{\epsilon}+24C\epsilon^{1/4}.
\end{align*}
Where we have used the triangle inequalities and then the bounds from \eqref{afterfirstunitary} and \eqref{A2A4U2} and then the bounds on $\lVert[A_2,A_3]\rVert$ and $\lVert\{A_3,A_4\}\rVert$. From the above two relations, we can deduce that $\lVert[O,Z]\rVert $ and $\lVert\{O,X\}\rVert$ also satisfy the same bound. Expanding the operator $O$ as $o_1\mathbb{1}+o_x X+o_y Y+o_zZ$, we can conclude the following
\begin{align*}
    &\lVert o_xX+o_yY\rVert \leq 2(C_2+\bar{C}_2+C_3+\bar{C}_3)\sqrt{\epsilon}+12C\epsilon^{1/4}, \\
    &\lVert o_1\mathbb{1}+o_xX\rVert\leq 2(C_{\gamma}+\bar{C}_3+\bar{C}_4)\sqrt{\epsilon}+12C\epsilon^{1/4}, 
\end{align*}
where from the 2nd relation we can write $\lvert o_1\pm o_x\rvert\leq 2(C_{\gamma}+\bar{C}_3+\bar{C}_4)\sqrt{\epsilon}+12C\epsilon^{1/4}$, which implies that $|o_1|,|o_x|$ are also bounded by the same number. Substituting this in the first inequality we can bound $\lVert o_y\rVert$ and then from there we can show that
\begin{align*}
    \lVert O-Z\rVert&\leq \lVert O-o_zZ\rVert+\lVert o_zZ-Z\rVert=\lVert o_1\mathbb{1}+o_xX+o_yY\rVert+\lVert o_zZ-Z\rVert\leq  \lVert o_1\mathbb{1}+ o_xX\rVert+\lVert o_yY\rVert +\lVert o_zZ-Z\rVert \\ &\leq  2(C_3+\bar{C}_3+2\bar{C}_4+2C_{\gamma}-2C_2)\sqrt{\epsilon}+12C\epsilon^{1/4},
\end{align*}
where we have used the triangle inequalities and then the bounds on respective operators proved just above. Also, $|o_z\pm1|\leq 2C_{3}\sqrt{\epsilon}$ because the minimum magnitude of $\hat{A}_3$ eigenvalues is $1-2C_{3}\sqrt{\epsilon}$. In the same way, we can get the following bound on  $\lVert Q-X\rVert$ 
\begin{align*}
& \lVert Q-X\rVert \leq 2(C_{\gamma}+C_2+C_3+2\bar{C}_2+\bar{C}_6)\sqrt{\epsilon}+12C\epsilon^{1/4}
\end{align*}
Substituting these relations in \eqref{afterfirstunitary}, and by using the  triangle inequalities, we get that 
\begin{align*}
    &\lVert U^{\dagger}\hat{A}_3U-X\otimes Z\rVert \leq \lVert U^{\dagger}\hat{A}_3U-X\otimes O\rVert+\lVert X\otimes O-X\otimes Z\rVert\leq 2(C_3+2\bar{C}_3+2\bar{C}_4+2C_{\gamma}-2C_2)\sqrt{\epsilon}+16C\epsilon^{1/4},\nonumber \\
    &\lVert U^{\dagger}\hat{A}_6U-Z\otimes X\rVert \leq \lVert U^{\dagger}\hat{A}_6U-Z\otimes Q\rVert+\lVert Z\otimes Q-Z\otimes X\rVert\leq 2(C_{\gamma}+C_2+C_3+2\bar{C}_2+2\bar{C}_6)\sqrt{\epsilon}+16C\epsilon^{1/4}.
\end{align*}
This completes the proof.
\end{proof}

\begin{lemma}\label{robuststatelemma}
    Suppose a non-maximal quantum violation $5-\epsilon$ of the inequality \eqref{TncIneq} is observed. Then, there exists a basis in the subspace $V$ such that 
    \begin{align}
        |\braket{\hat{\psi}|\phi^+}|^2\geq 1-\left(s_1\epsilon+s_2\epsilon^{3/4}+s_3\sqrt{\epsilon}\right),
    \end{align}
    where $s_1,s_2$ are positive constants, and $\ket{\phi^+}$ is the two qubit maximally entangled state and $U(P|\ket{\psi})=\ket{\hat\psi}$.
\end{lemma}
\begin{proof}
Let us consider the following expectation value in any arbitrary state from $V$,
\begin{align*}
    \braket{X\otimes X}&= 1-\frac{1}{2}\lVert X\otimes \mathbb{1}-\mathbb{1}\otimes X\rVert^2
    \geq 1-\frac{1}{2}\left(\lVert X\otimes \mathbb{1}-A_1\rVert^2+\lVert A_1-A_4\rVert^2+\lVert A_4-\mathbb{1}\otimes X\rVert^2\right),\\
    &\geq 1-2\left[4C_4^2+c_4^2\right]\epsilon-32c_4 C \epsilon^{3/4}-128 C^2 \sqrt{\epsilon},
\end{align*}
where we have used triangle inequality in the first line, then used Eq. \eqref{srcbound3} and Lemma \ref{robustoperatorslemma} in the second line. In the same way by using Lemma \ref{robustoperatorslemma} and Eqs. \eqref{srcbound4} and \eqref{srcbound5}, we can show respectively that 
    \begin{align}
        \braket{Z\otimes Z}&\geq 1-2\left[9C_\gamma^2+C_2^2+c_2^2+C_5^2\right]\epsilon-4(3C_\gamma +4c_2) C\epsilon^{3/4} -34 C^2\sqrt{\epsilon}\nonumber \\
       - \braket{Y\otimes Y}&= 1-\frac{1}{2}\lVert X\otimes Z-Z\otimes X\rVert^2
    \geq 1-\frac{1}{2}\left(\lVert X\otimes Z-A_3\rVert^2+\lVert A_3-A_6\rVert^2+\lVert A_6-Z\otimes X\rVert^2\right),\nonumber\\
        &\geq 1-2\left[c_3^2+c_6^2+C_3^2+C_6^2\right]\epsilon-32 \left(c_3+c_6\right)C \epsilon^{3/4} -256 C^2\sqrt{\epsilon},\label{yystate}
    \end{align}
    where we have used ${\rm i} XZ=- {\rm i}ZX=Y$ with ${\rm i}=\sqrt{-1}$. 
From the above three relations, it is straightforward to guess that the state must be of the form $\ket{\hat\psi}=\cos(\theta)\ket{\phi^+}+\sin(\theta)\ket{\phi^-}$, where $\ket{\phi^-}$ is orthogonal to $\ket{\phi^+}$. The lowest bound on the fidelity of $\ket{\hat\psi}$ with $\ket{\phi^+}$ comes from the last relation \eqref{yystate}. Therefore, we get 
\begin{align*}
        \cos^2(\theta)-\sin^2(\theta)\geq 1-2\left(s_1\epsilon+s_2\epsilon^{3/4}+s_3\sqrt{\epsilon}\right),
\end{align*}
where $s_1,s_2,s_3$ denote the constant factors (divided by $2$) in front of $\epsilon, \epsilon^{3/4},\sqrt{\epsilon}$ respectively in the Eq. \eqref{yystate}. This immediately implies the following
\begin{align*}
    |\braket{\hat\psi|\phi^+}|^2=\cos^2(\theta)\geq 1-\left(s_1\epsilon+s_2\epsilon^{3/4}+s_3\sqrt{\epsilon}\right).
    \end{align*}
    Hence, we prove the lemma.
\end{proof}
\end{document}